\documentclass[
longbibliography,
reprint,
amsmath,amssymb,
aps,
pra,
superscriptaddress,
nofootinbib
]{revtex4-2}
\usepackage{graphicx}
\usepackage{amsmath,amssymb,amsthm}
\usepackage{bm}
\usepackage{physics}
\usepackage[colorlinks,linkcolor=blue,citecolor=blue]{hyperref}
\usepackage{mathtools}
\theoremstyle{plain}
\newtheorem{theorem}{Theorem}
\newtheorem{lemma}[theorem]{Lemma}
\newtheorem{proposition}[theorem]{Proposition}
\newcommand{\id}{\mathrm{id}}
\DeclareMathOperator*{\argmax}{arg~max}
\newcommand{\wt}{\mathrm{wt}}
\newcommand{\ii}{\mathrm{i}}
\newcommand{\ff}{\mathbb{F}_2}
\newcommand{\ee}{\mathbf{e}}
\newcommand{\sss}{\mathbf{s}}
\newcommand{\ccc}{\mathbf{c}}
\newcommand{\uu}{\mathbf{u}}
\newcommand{\vv}{\mathbf{v}}
\newcommand{\cc}{\mathbb{C}}
\newcommand{\yy}{\mathbf{y}}
\newcommand{\xx}{\mathbf{x}}
\newcommand{\bb}{\mathbf{b}}
\newcommand{\bmalpha}{\bm{\alpha}}
\newcommand{\bmbeta}{\bm{\beta}}
\newcommand{\bmmu}{\bm{\mu}}
\newcommand{\one}{\mathbf{1}}
\newcommand{\onegamma}{\mathbf{1}_{\gamma}}
\newcommand{\bmxi}{\bm{\xi}}
\newcommand{\pauli}{\mathcal{P}}
\newcommand{\calH}{\mathcal{H}}
\newcommand{\calG}{\mathcal{G}}

\begin{document}

\title{
Proof of a positive coherent-error threshold for topological quantum codes
}

\author{Shiro Tamiya}
\email{shiro.tamiya01@gmail.com}
\affiliation{Photon Science Center, Graduate School of Engineering, The University of Tokyo, 7-3-1 Hongo, Bunkyo-ku, Tokyo 113-8656, Japan}
\affiliation{Department of Applied Physics, Graduate School of Engineering, The University of Tokyo, 7-3-1 Hongo, Bunkyo-ku, Tokyo 113-8656, Japan}

\author{Masato Koashi}
\email{koashi@qi.t.u-tokyo.ac.jp}
\affiliation{Photon Science Center, Graduate School of Engineering, The University of Tokyo, 7-3-1 Hongo, Bunkyo-ku, Tokyo 113-8656, Japan}
\affiliation{Department of Applied Physics, Graduate School of Engineering, The University of Tokyo, 7-3-1 Hongo, Bunkyo-ku, Tokyo 113-8656, Japan}

\begin{abstract}
Threshold analyses of quantum error-correcting codes are well established for stochastic error models, in which errors occur randomly with given probabilities.
However, errors in actual devices can also be coherent, such as unwanted $Z$ rotations due to imperfect control, which are not captured by stochastic error models.
For the surface code, numerical studies have indicated threshold behavior even under coherent errors, but a rigorous proof of threshold existence is lacking.
Here we prove that a positive threshold for coherent $Z$-rotation errors exists for quantum low-density parity-check codes with a bounded number of logical qubits, including the surface code and other topological codes.
Specifically, we show that the maximum-likelihood Pauli recovery suppresses the entanglement infidelity exponentially in the code distance up to a prefactor linear in the number of physical qubits whenever the rotation angles lie below a constant that is independent of the code size.
The proof combines Fourier analysis to retain the interference among the amplitudes of coherent errors with the cluster expansion of abstract polymer models.
Our results expand the theoretical foundation of quantum error correction and offer a statistical-mechanical description of quantum error correction beyond stochastic errors.
\end{abstract}
\maketitle

\section{Introduction}

Quantum error correction (QEC) protects quantum information by encoding
logical qubits into a larger Hilbert space of many physical qubits, and constitutes a key ingredient of fault-tolerant quantum computation (FTQC)~\cite{Shor1995Scheme,Steane1996Simple,Knill2000Theory,Terhal2015Quantum}.
The central guarantee of QEC is the existence of an error-correcting threshold, which states that the logical error rate can be suppressed to an arbitrarily small level by growing the code distance, provided the strength of the error is below a nonzero threshold value~\cite{Dennis_2002, Terhal2015Quantum}.
Threshold analyses have been developed primarily for stochastic error models, in which errors occur randomly with given probabilities~\cite{Dennis_2002,Fowler2012Proof,yoshida2026proof,Kubica_2019}.

Errors in quantum devices, however, are not necessarily stochastic.
A prominent class of errors beyond stochastic error models consists of coherent errors caused by imperfect control and miscalibration~\cite{Greenbaum2018Modeling,Kueng2016Comparing, Bravyi2018Correcting}.
In this work, we consider a product of single-qubit $Z$ rotation errors on all physical qubits as a coherent error model for QEC~\cite{Bravyi2018Correcting,VennBeri2020Error}.
Each single-qubit rotation decomposes as $e^{\ii \theta Z}=\cos(\theta)I+\ii \sin(\theta)Z$, so expanding the product over the $n$ physical qubits expresses the error as a sum of exponentially many terms, one for each pattern of phase flips on the qubits.
The coefficient of each term is a complex number whose absolute value and phase are both fixed by the rotation angles.
We call these coefficients the amplitudes of the corresponding error patterns.

For stochastic errors, the probability of each outcome, that is, of each syndrome and each effect on the encoded information, is the sum of the probabilities of the corresponding patterns.
These probabilities are nonnegative real numbers, so they cannot cancel, and each pattern contributes to the sum independently of the other patterns.
In contrast, under coherent errors, the amplitudes of distinct patterns that produce the same syndrome and the same effect on the encoded information are first added, with the squared absolute value of the sum giving the probability of that outcome.
The difference between the squared absolute value of the sum and the sum of the squared absolute values of the amplitudes, which can have either sign, indicates the effect of the interference: the patterns can partially cancel or reinforce one another depending on their relative phases~\cite{Bravyi2018Correcting,Behrends2025TheSurface,Iverson2020Coherence}.
The codespace of an error-correcting code determines which amplitudes combine and preserving the cancellations among the amplitudes is therefore required for the threshold analysis under coherent errors.

One way to sidestep this requirement is to remove the interference.
Coherent errors can be transformed into effectively stochastic Pauli errors through random Pauli operations in twirling~\cite{Geller2013Efficient,WallmanEmerson2016Noise} or through the random measurement outcomes intrinsic to teleportation-based error correction~\cite{Chang2025Taming}. 
Whether QEC can nevertheless protect quantum information against coherent errors directly, without making coherent errors effectively stochastic, is a fundamental question for the theory of QEC.

We study this question for topological stabilizer codes defined on geometrically local lattices~\cite{Bravyi1998surface,Bombin2007Optimal,Tomita2014LowDistance, Bravyi2018Correcting,Fujii2015Topological,Kitaev2003anyons, Dennis_2002,Bombin2007Topological,Bombin2006Topological}.
The geometric locality eases the implementation of QEC, so topological codes, and the surface code in particular, are among the leading candidates for realizing QEC.
These codes have been demonstrated in experiments~\cite{Krinner_2022,Youwei2022,Bluvstein2025,Google2023Suppressing,Google2024Quantum,Berthusen_2024,Lacroix_2025,Sales_Rodriguez_2025}.
Numerical studies indicate that there is a nonzero threshold under coherent errors for the surface code and other stabilizer codes~\cite{Suzuki2017Efficient,Bravyi2018Correcting,VennBeri2020Error,Venn2023Coherent,MartonAsboth2023Coherent,BaoAnand2024Phases,Behrends2025TheSurface,Pato2025Logical,Behrends_2025,yan2026nonlinearsigmamodelsurface}.
Analytical studies also indicate that QEC suppresses coherent errors, but for topological codes the bound has code-dependent prefactors~\cite{Huang2019Performance} that may grow with the code size or require the rotation angle to vanish as the code size grows~\cite{Iverson2020Coherence}.
Thus, to our knowledge, a proof of a nonzero coherent-error threshold at a fixed rotation angle has been lacking for the surface code and other topological codes, even when the syndrome measurement is noiseless.
This leaves the mechanism of coherent-error suppression not fully understood.

In this work, we prove the existence of a positive threshold for coherent $Z$-rotation errors in the code-capacity setting, in which the errors act only on the physical qubits and the syndrome measurement is performed in a noiseless way.
After the syndrome is measured, we apply the maximum-likelihood (ML) Pauli recovery~\cite{BaoAnand2024Phases,Darmawan_2017,Darmawan2024Low,Bravyi_2014}.
The codes are Calderbank–Shor–Steane (CSS) quantum low-density parity-check (QLDPC) codes with a bounded number of logical qubits, in which every stabilizer generator acts on a bounded number of physical qubits and every physical qubit belongs to a bounded number of stabilizer generators~\cite{MacKay2004Sparse,Breuckmann2021Quantum}.
 This class includes the surface code and other topological codes.
 For such a family, we prove that the entanglement infidelity of the recovered logical information is suppressed exponentially in the code distance, up to a prefactor linear in the number of physical qubits, whenever the absolute values of the rotation angles are below a constant that is independent of the code size.
 If the code distance grows superlogarithmically in the number of physical qubits, the infidelity vanishes as the code size grows, and the constant is a rigorous lower bound on the threshold.

The main obstacle to such a proof is that replacing the amplitudes of the error patterns by absolute values, which discards the cancellations among the amplitudes due to interference, may give a pessimistic bound on the infidelity.
We instead bound the infidelity using pairwise products of amplitude sums over error patterns that share a syndrome but have different effects on the encoded information, and use Fourier analysis~\cite{ODonnell2014Analysis} to represent these products through partition functions of an abstract polymer model, whose polymers are connected subsets of qubits that produce a trivial syndrome.
The cluster expansion~\cite{Koteck__1986_polymer,friedli_velenik_2017_polymer,Helmuth_2019,Mann_2024cluster} then expresses the logarithm of each partition function as a sum over clusters, i.e., connected groups of polymers, which converges for QLDPC codes whenever the absolute values of the rotation angles are below a constant that is independent of the code size.
To bound the infidelity, we compare the sums over clusters for two codewords, which differ only in the contributions from clusters with a nontrivial effect on the encoded information.
The number of qubits in every such cluster is at least the code distance, which gives the exponential bound on the infidelity stated above.
Fourier analysis and the cluster expansion have both been used in classical coding theory~\cite{Forney2011CodesOnGraphs,Kudekar2011Decay,Macris2012Polymer}.

Consequently, we establish a coherent-error threshold in the code-capacity setting for topological quantum codes, thereby narrowing the gap between the provable threshold and the numerically estimated value to a quantitative, rather than a qualitative, problem.
Our contribution is the proof that the threshold is nonzero with an explicit lower bound, and the understanding this proof provides regarding how coherent errors are suppressed.
Our results and proof technique may serve as a key step toward a threshold theorem for FTQC with topological codes that handles circuit-level coherent errors induced by faulty syndrome measurement~\cite{Fowler2012Proof,Gottesman2014constant,Tamiya2026Fault,Christandl2025Fault}.
Beyond its relevance to FTQC, our approach may also offer a statistical-mechanical perspective on why coherent errors are suppressed.
The correspondence between error-correcting thresholds and phase transitions in disordered classical spin models has been investigated mainly for stochastic Pauli errors~\cite{Dennis_2002,Wang_2003,ChubbFlammia2021Statistical,Kubica2018Three}, and we leave the connection to our polymer representation for future work.

The rest of this paper is organized as follows. 
In Sec.~\ref{sec:preliminaries}, we fix the notation and introduce stabilizer codes, CSS codes, the syndrome subspaces of CSS codes, QLDPC codes, the entanglement fidelity, and Fourier analysis on binary vector spaces. 
In Sec.~\ref{sec:setting}, we define the coherent $Z$-rotation error model and the ML Pauli recovery. In Sec.~\ref{sec:main-results} we state the main theorem and discuss its consequences. 
In Sec.~\ref{sec:fidelity-ml}, we derive an exact expression for the entanglement fidelity under the ML recovery by decomposing the coherent-error amplitudes according to syndrome and logical class, that is, the effect on the encoded information.
In Sec.~\ref{sec:fourier-rep}, we use Fourier analysis to reduce the resulting infidelity estimate to a Fourier $1$-norm bound that retains the interference among the amplitudes.
In Sec.~\ref{sec:polymer}, we reformulate this Fourier representation as an abstract polymer model and derive the estimates needed for the cluster expansion.
In Sec.~\ref{sec:cluster}, we combine the cluster expansion with the code-distance constraint to obtain the exponential bound and complete the proof of the main theorem.
In Sec.~\ref{sec:conclusion}, we conclude with a summary and future directions. 

\section{Preliminaries}
\label{sec:preliminaries}
In this section, we fix notation and definitions.
For each positive integer $n$, define $[n]\coloneqq \{1,\ldots,n\}$.
Let $\mathbb{R}$ and $\mathbb{C}$ be the sets of real and complex numbers, respectively.
For a condition $J$, let $\bm 1[J]$ equal one if $J$ holds and zero otherwise.
The set of linear operators on a Hilbert space $\calH$ is denoted by $\mathcal{L}(\calH)$.
For an operator $A$ on $\calH$, $A^\dagger$ denotes its adjoint, and $\Tr A$ its trace.
In Secs.~\ref{subsec:prelim-stabilizer} and~\ref{subsec:prelim-css}, we give the notation and the definitions of stabilizer codes, CSS codes, and QLDPC codes.
In Sec.~\ref{subsec:prelim-syndrome-bases}, we introduce the syndrome subspaces of CSS codes.
In Sec.~\ref{subsec:prelim-fidelity}, we introduce the entanglement fidelity.
In Sec.~\ref{subsec:prelim-fourier}, we introduce Fourier analysis over binary vector spaces.
\subsection{Stabilizer codes}
\label{subsec:prelim-stabilizer}
Let $\mathbb{F}_2$ be the finite field with elements $\{0,1\}$, and let $\mathbb{F}_2^n$ be the $n$-dimensional row vector space over $\mathbb{F}_2$.
For $\mathbf{x},\mathbf{y}\in\ff^n$, define $\mathbf{x}\cdot\mathbf{y}\coloneqq\sum_{i=1}^n x_iy_i\bmod 2$.
For a row vector $\mathbf{x}\in\mathbb{F}_2^n$, define the support of $\xx$ by
\begin{equation}
\label{eq:support-def}
    \mathrm{supp}(\xx)\coloneqq \{j\in[n]\colon x_j=1\},
\end{equation}
and define its Hamming weight by
\begin{equation}
    \wt(\mathbf{x})\coloneqq |\mathrm{supp}(\xx)|.
\label{eq:hamming weight}
\end{equation}

Let $\{\ket{0},\ket{1}\}$ be the computational basis for the single-qubit Hilbert space $\mathbb{C}^2$, and let $\qty(\mathbb{C}^2)^{\otimes n}$ be the $n$-qubit Hilbert space.
For $\bb=\qty(b_1,\ldots,b_n)\in\ff^n$, define the computational-basis state
\begin{equation}
    \ket{\mathbf b}\coloneqq\bigotimes_{j=1}^n\ket{b_j}.
\end{equation}
Let $X,Y$, and $Z$ be the single-qubit Pauli operators defined by $X\coloneqq \ketbra{0}{1}+\ketbra{1}{0}$, $Y\coloneqq \ii \ketbra{1}{0}-\ii \ketbra{0}{1}$, and $Z\coloneqq \ketbra{0}{0}-\ketbra{1}{1}$, where $\ii\coloneqq \sqrt{-1}$.
Let $I\coloneqq \ketbra{0}{0}+\ketbra{1}{1}$ be the identity operator.
The Pauli group on $n$ qubits, denoted by $\pauli_n$, is the set of operators $P=\alpha\bigotimes_{i=1}^n P_i$, where $\alpha\in \{\pm1, \pm\ii\}$ and $P_i\in\{I, X,Y,Z\}$.

A \textit{stabilizer} $\mathcal{S}$ is an abelian subgroup of $\mathcal{P}_n$ not containing $-I^{\otimes n}$.
A \textit{stabilizer code} $\mathcal{Q}$ is a linear subspace of the $n$-qubit Hilbert space given by the common $+1$-eigenspace for all elements in $\mathcal{S}$, i.e., $\mathcal{Q}\coloneqq \{\ket{\psi}\colon s\ket{\psi}=\ket{\psi}~\text{for all~$s\in\mathcal{S}$}\}$, also called the \textit{codespace}~\cite{Gottesman1997Stabilizer, gottesman2009introductionquantumerrorcorrection}.
If $\mathcal{S}$ is generated by $n-k$ independent elements, referred to as \textit{stabilizer generators}, the stabilizer code $\mathcal{Q}$ has dimension $2^k$.
We say the code $\mathcal{Q}$ encodes $k$ logical qubits.

The normalizer of $\mathcal{S}$, denoted $\mathcal{N}(\mathcal{S})$, is given by $\mathcal{N}(\mathcal{S})\coloneqq \{P\in\mathcal{P}_n\colon P\mathcal{S}P^{-1}=\mathcal{S}\}$, where $P\mathcal{S}P^{-1}\coloneqq \{PsP^{-1}\colon s\in\mathcal{S}\}$.
The set of elements that commute with every element of $\mathcal{P}_n$ is called the center of $\mathcal{P}_n$, denoted $\mathcal{Z}\coloneqq \{\alpha I^{\otimes n}\colon\alpha\in\{\pm1, \pm \mathrm{i}\}\}$.
The logical Pauli group is defined as the quotient group $\mathcal{N}(\mathcal{S})/(\mathcal{S}\cdot \mathcal{Z})$, with $\mathcal{S}\cdot \mathcal{Z}=\{sz\colon s\in\mathcal{S}, z\in\mathcal{Z}\}$.

An operator $L\in\mathcal{N}(\mathcal{S})$ is referred to as a \textit{logical operator}, and is called \textit{nontrivial} if its coset $[L]$ is nontrivial, i.e., $[L]\neq [I]$.
The code distance $d$ is given by $d\coloneqq \min\{|L|\colon L\in \mathcal{N}(\mathcal{S})\setminus(\mathcal{S}\cdot \mathcal{Z})\}$, where the weight $|P|$ of a Pauli operator $P=\alpha\bigotimes_{i=1}^n P_i$ denotes the number of physical qubits on which it acts nontrivially, i.e., $|P|\coloneqq\abs{\{i\in\{1,\ldots, n\}\colon P_i\neq I\}}$.
A stabilizer code with $n$ physical qubits, $k$ logical qubits, and code distance $d$ is referred to as an $[[n, k, d]]$ code.
Throughout this paper, we refer to the $n$ physical qubits constituting the stabilizer code as \emph{data qubits}.

\subsection{CSS codes and QLDPC codes}
\label{subsec:prelim-css}
A \textit{Calderbank-Shor-Steane (CSS) code} is a stabilizer code whose stabilizer generators are composed entirely of tensor products of $I$ and $X$ ($X$-type generators) or $I$ and $Z$ ($Z$-type generators)~\cite{Calderbank_1996, Steane1996,Steane1996Simple}.
For binary row vectors $\mathbf{a}, \mathbf{b}\in\ff^n$, we define
\begin{equation}
    X(\mathbf{a})\coloneqq \prod_{j=1}^n X^{a_j}_j,\quad Z({\mathbf{b}})\coloneqq \prod_{j=1}^n Z^{b_j}_j,
\end{equation}
where $X_j^0=Z_j^0\coloneqq I_j$.

Fix linearly independent $X$-type stabilizer generators $g_X^{(i)}\coloneqq X({\mathbf{h}_X^{(i)}})$ for $i\in[m_X]$ and linearly independent $Z$-type generators $g_Z^{(i)}\coloneqq Z({\mathbf{h}_Z^{(i)}})$ for $i\in[m_Z]$, specified by binary row vectors $\mathbf{h}_X^{(i)},\mathbf{h}_Z^{(i)}\in\ff^n$.
Collecting these vectors as rows gives binary parity-check matrices $H_X\in\mathbb{F}_2^{m_X\times n}$ and $H_Z\in\mathbb{F}_2^{m_Z\times n}$, whose $i$-th rows are $\mathbf{h}_X^{(i)}$ and $\mathbf{h}_Z^{(i)}$, both of full row rank.
This causes no loss of generality, because discarding linearly dependent rows leaves the code unchanged and does not increase the row or column weights.
The stabilizer of the CSS code is generated by these $m_X+m_Z$ operators, $\mathcal{S}=\langle g_X^{(1)},\ldots,g_X^{(m_X)},g_Z^{(1)},\ldots,g_Z^{(m_Z)}\rangle$.
The number of logical qubits is
\begin{equation}
    \label{eq:num-logical-def}
    k=n-m_X-m_Z,
\end{equation}
and the dimension of $\mathcal{Q}$ is $2^k$.
The condition that all generators commute is equivalent to $H_ZH_X^\top=0$.

Only $Z$-type errors are considered in this paper, so we define the syndrome and the logical class for $Z$-type errors as follows.
The \textit{syndrome} $\sss\in\ff^{m_X}$ of the $Z$-type error $Z({\mathbf{e}})$, obtained by measuring the $X$-type stabilizer generators, is given by
\begin{equation}
\sss\coloneqq\mathbf{e}H_X^\top \in\ff^{m_X}.    
\end{equation}
We choose and fix binary matrices $O_X,O_Z\in\ff^{k\times n}$ whose rows represent independent logical $X$ and logical $Z$ operators, respectively~\cite{Gottesman1997Stabilizer}.
The $(i,j)$-th entry of $O_X$ is one when the $i$-th logical $X$ operator acts by $X$ on data qubit $j$, and $O_Z$ is defined in the same way from the logical $Z$ operators.
Since the logical operators commute with stabilizer generators and are chosen so that the $i$-th logical $X$ operator anticommutes with the $i$-th logical $Z$ operator and commutes with the others, the matrices satisfy the conditions
\begin{align}
&O_ZH_X^{\top}=0,\\
&O_XH_Z^{\top} = 0,\\
&O_X O_Z^\top=I_k\label{eq:logical-O},
\end{align}
where $I_k$ denotes the $k\times k$ identity matrix.
For $\ee\in\ff^n$, the \textit{logical class} of the $Z$-type error $Z(\ee)$ is given by
\begin{equation}
\label{eq:logical-class-def}
\uu\coloneqq\ee O_X^\top\in\ff^k.
\end{equation}

The \textit{logical $Z$ basis} of $\mathcal Q$ is the orthonormal basis $\{\ket{\overline{\bmalpha}}_Z\}_{\bmalpha\in\ff^k}$, where
\begin{align}
\label{eq:logical-Z-basis}
\ket{\overline{\bmalpha}}_Z
&\coloneqq 2^{-m_X/2}\sum_{\yy\in\ff^{m_X}}
\ket{\bmalpha O_X+\yy H_X}.
\end{align}
These $2^k$ states are the \textit{codewords} of the CSS code~\cite{Gottesman1997Stabilizer,gottesman2009introductionquantumerrorcorrection}.
A $Z$-type operator $Z(\ee)$ is a logical $Z$ operator if and only if $\ee H_X^\top=\bm0$.
For $\uu\in\ff^k$, define the logical $Z$ operator representing $\uu$ by
\begin{equation}
    \overline{Z}(\uu)\coloneqq Z(\uu O_Z).
\label{eq:Zbar-def}
\end{equation}
The operator $\overline Z(\uu)$ acts on the logical $Z$ basis as
\begin{equation}
\label{eq:logical-Z-codeword-action}
    \overline Z(\uu)\ket{\overline{\bmalpha}}_Z
    =(-1)^{\uu\cdot\bmalpha}\ket{\overline{\bmalpha}}_Z.
\end{equation}
The \textit{logical $X$ basis} of $\mathcal Q$ is the orthonormal basis $\{\ket{\overline{\mathbf p}}_X\}_{\mathbf p\in\ff^k}$, where
\begin{equation}
\label{eq:logical-X-codeword-basis}
    \ket{\overline{\mathbf p}}_X
    \coloneqq 2^{-k/2}\sum_{\bmalpha\in\ff^k}
    (-1)^{\bmalpha\cdot\mathbf p}\ket{\overline{\bmalpha}}_Z.
\end{equation}

The vectors $\ee$ for which $Z(\ee)$ is a logical $Z$ operator are
\begin{align}
\label{eq:kerHX-decomposition}
&\{\ee\in\ff^n\colon\ee H_X^\top=\bm0\}\nonumber\\
&\quad=\{\uu O_Z+\mathbf aH_Z\colon\uu\in\ff^k,\ \mathbf a\in\ff^{m_Z}\},
\end{align}
where $\uu O_Z+\mathbf aH_Z$ has logical class $\uu$.
Thus $Z(\ee)$ with $\ee H_X^\top=\bm0$ is in $\mathcal S$ if and only if $\ee O_X^\top=\bm0$, and the minimum weight $d_Z$ of a nontrivial logical $Z$ operator is
\begin{equation}
\label{eq:dz-def}
    d_Z=\min\{\wt(\mathbf{e})\colon \mathbf{e}H_X^\top=0,~~\mathbf{e}O_X^\top\neq 0\}.
\end{equation}
Since the code distance $d$ is the minimum weight of a nontrivial logical Pauli operator, $d\le d_Z$.

Finally, a family of CSS codes is called $(r,c)$-\textit{QLDPC} if, for every code in the family, the row and column weights of the parity-check matrices $H_X$ and $H_Z$ are bounded by constants $r$ and $c$, respectively.

\subsection{Syndrome subspaces}
\label{subsec:prelim-syndrome-bases}
For $\sss\in\ff^{m_X}$ and $\bmalpha\in\ff^k$, the state
\begin{equation}
\label{eq:syndrome-logical-basis}
\ket{\sss,\bmalpha}_Z
\coloneqq 2^{-m_X/2}\sum_{\yy\in\ff^{m_X}}(-1)^{\sss\cdot\yy}
\ket{\bmalpha O_X+\yy H_X}
\end{equation}
has the same computational-basis components as the codeword $\ket{\overline{\bmalpha}}_Z=\ket{\bm0,\bmalpha}_Z$, with the signs $(-1)^{\sss\cdot\yy}$, and satisfies
\begin{align}
 g_X^{(i)}\ket{\sss,\bmalpha}_Z&=(-1)^{s_i}\ket{\sss,\bmalpha}_Z,\\
  g_Z^{(j)}\ket{\sss,\bmalpha}_Z&=\ket{\sss,\bmalpha}_Z,\\
 \overline Z(\uu)\ket{\sss,\bmalpha}_Z&=(-1)^{\uu\cdot\bmalpha}\ket{\sss,\bmalpha}_Z.
\label{eq:logical-basis-action}
\end{align}
The states $\ket{\sss,\bmalpha}_Z$, over all $(\sss,\bmalpha)$, form an orthonormal basis of the subspace with eigenvalue $+1$ of every $g_Z^{(j)}$.

For each $\sss$, define the projector onto the subspace with eigenvalue $(-1)^{s_i}$ of every $g_X^{(i)}$ and eigenvalue $+1$ of every $g_Z^{(j)}$ by
\begin{equation}
\label{eq:projector-basis}
\Pi_\sss\coloneqq\sum_{\bmalpha\in\ff^k}\ket{\sss,\bmalpha}_Z{}_Z\!\bra{\sss,\bmalpha}.
\end{equation}
By the orthonormality of the states $\ket{\sss,\bmalpha}_Z$, $\Pi_\sss^\dagger=\Pi_\sss$ and $\Pi_\sss\Pi_{\sss'}=\bm1[\sss=\sss']\Pi_\sss$, and $\Pi_\mathcal Q=\Pi_{\bm0}$ is the projector onto the codespace.
We call the image of $\Pi_\sss$ the \textit{syndrome subspace} associated with syndrome $\sss$. 
Summing over the syndrome gives
\begin{equation}
\label{eq:syndrome-completeness}
\sum_{\sss\in\ff^{m_X}}\Pi_\sss
=\sum_{\substack{\bb\in\ff^n\\\bb H_Z^\top=\bm0}}\ket{\bb}\bra{\bb}.
\end{equation}

For $\mathbf p\in\ff^k$, define
\begin{equation}
\label{eq:logical-X-basis}
\ket{\sss,\mathbf p}_X
\coloneqq 2^{-k/2}\sum_{\bmalpha\in\ff^k}(-1)^{\bmalpha\cdot\mathbf p}\ket{\sss,\bmalpha}_Z.
\end{equation}
These states form another orthonormal basis of the syndrome subspace.
In particular,  $\ket{\bm0,\mathbf p}_X=\ket{\overline{\mathbf p}}_X$.
A $Z$-type error $Z(\ee)$ with syndrome $\sss=\ee H_X^\top$ and logical class $\uu=\ee O_X^\top$ acts on the codewords as
\begin{align}
Z(\ee)\ket{\overline{\bmalpha}}_Z
&=(-1)^{\bmalpha\cdot\uu}\ket{\sss,\bmalpha}_Z,\\
Z(\ee)\ket{\overline{\mathbf p}}_X
&=\ket{\sss,\mathbf p+\uu}_X.
\label{eq:Pauli-basis-transition}
\end{align}
Thus $Z(\ee)$ maps the codespace onto the subspace of syndrome $\sss$, preserves $\bmalpha$ up to the sign $(-1)^{\bmalpha\cdot\uu}$ of $\overline Z(\uu)$, and shifts $\mathbf p$ by $\uu$.

\subsection{Entanglement fidelity and diamond norm}
\label{subsec:prelim-fidelity}
A \textit{quantum channel} on a finite-dimensional Hilbert space $\calH$ is a completely positive and trace-preserving linear map $\mathcal{E}\colon\mathcal{L}(\calH)\to\mathcal{L}(\calH)$~\cite{Watrous_2018}.
A linear map $\mathcal{E}\colon\mathcal{L}(\calH)\to\mathcal{L}(\calH)$ is a quantum channel if and only if it admits a \textit{Kraus representation}
\begin{equation}
    \mathcal{E}(X)=\sum_\alpha K_\alpha X K_\alpha^\dagger,
\end{equation}
for all $X\in\mathcal{L}(\calH)$, where the operators $K_\alpha$ satisfy
\begin{equation}
    \sum_\alpha K_\alpha^\dagger K_\alpha=I_{\calH}.
\end{equation}
The operators $K_\alpha$ are called \textit{Kraus operators}, and $I_{\calH}$ is the identity operator on $\calH$.

Let $\Lambda$ be a quantum channel on a $q$-dimensional Hilbert space $\calH_A$, fix an orthonormal basis $\{\ket{a}\}_{a=1}^q$, and let
\begin{equation}
    \ket{\Phi_q}\coloneqq \frac{1}{\sqrt{q}}\sum_{a=1}^q \ket{a}_R\ket{a}_A
\end{equation}
be a maximally entangled state with a reference system $R$ whose Hilbert space $\calH_R$ has the same dimension $q$ as $\calH_A$.
Throughout this work, $F_e(\Lambda)$ denotes the entanglement fidelity of the maximally mixed state $I_{\calH_A}/q$ under the channel $\Lambda$, defined as
\begin{equation}
    F_e(\Lambda)\coloneqq \bra{\Phi_q}\qty(\id_R\otimes \Lambda)\qty(\ketbra{\Phi_q}{\Phi_q})\ket{\Phi_q}.
\end{equation}
Here $\id_R$ denotes the identity map on $\mathcal{L}(\calH_R)$.
For any family of Kraus operators $K_\alpha$ of $\Lambda$, the entanglement fidelity is given by~\cite{Schumacher1996Sending}
\begin{equation}
\label{eq:Fe-kraus}
    F_e(\Lambda)=\frac{1}{q^2}\sum_{\alpha}\abs{\Tr K_\alpha}^2.
\end{equation}

For a linear map
$\Phi\colon\mathcal{L}(\calH_A)\to\mathcal{L}(\calH_A)$,
the diamond norm is defined by~\cite{Watrous_2018}
\begin{equation}
\label{eq:diamond norm}
    \|\Phi\|_\diamond
    \coloneqq
    \max_{\substack{
        X\in\mathcal{L}(\calH_R\otimes\calH_A)\\
        \|X\|_1=1
    }}
    \left\|
        \qty(\id_R\otimes\Phi)\qty(X)
    \right\|_1,
\end{equation}
where $\|X\|_1\coloneqq\Tr\sqrt{X^\dagger X}$ is the trace norm.
For a quantum channel $\Lambda$, the quantity $\frac{1}{2}\norm{\Lambda-\id}_\diamond$ satisfies the bound~\cite{majenz2018entropyquantuminformationtheory}
\begin{equation}
\label{eq:diamond-Fe}
    \frac{1}{2}\|\Lambda-\id\|_\diamond\leq q\sqrt{1-F_e(\Lambda)}.
\end{equation}

\subsection{Fourier analysis on binary vector spaces}
\label{subsec:prelim-fourier}
In this subsection, we introduce the background on Fourier analysis on $\ff^N$~\cite{ODonnell2014Analysis}.
For functions $f,g\colon\ff^N\to\mathbb C$, define the inner product
\begin{equation}
\langle f,g\rangle\coloneqq 2^{-N}\sum_{\mathbf{y}\in\ff^N}\overline{f(\mathbf{y})}g(\mathbf{y}),
\label{eq:inner-product}
\end{equation}
where $\overline{(\cdot)}$ denotes complex conjugation.
The \emph{character} $\chi_{\bmxi}\colon \ff^N\rightarrow \{\pm 1\}$ indexed by $\bmxi\in\ff^{N}$ is defined by
\begin{equation}
\label{eq:character}
\chi_{\bmxi}(\mathbf{y})\coloneqq(-1)^{\bmxi\cdot\mathbf{y}}.
\end{equation}
The characters are multiplicative in the argument and in the index,
\begin{align}
\chi_{\bmxi}(\mathbf{y}+\mathbf{y}')&=\chi_{\bmxi}(\mathbf{y})\chi_{\bmxi}(\mathbf{y}'),\label{eq:char-mult-argument}\\
\chi_{\bmxi}(\yy)\chi_{\bmxi^\prime}(\yy)&=\chi_{\bmxi+\bmxi^\prime}(\yy)\label{eq:char-mult-index},
\end{align}
for all $\bmxi,\bmxi',\mathbf{y},\mathbf{y}'\in\ff^{N}$.
They also satisfy
\begin{equation}
2^{-N}\sum_{\mathbf{y}\in\ff^N}\chi_{\bmxi}(\mathbf{y})=\bm 1[\bmxi=\bm 0]
\label{eq:parity-projector}
\end{equation}
for all $\bmxi\in\ff^N$.
The characters form an orthonormal basis of the space of functions from $\ff^N$ to $\cc$ for the inner product of Eq.~\eqref{eq:inner-product}.

The \emph{Fourier transform} of $f$ is the function $\widehat{f}\colon\ff^N\to\cc$ with
\begin{align}
\widehat f(\bmxi)\coloneqq\langle\chi_{\bmxi},f\rangle=2^{-N}\sum_{\mathbf{y}\in\ff^N}\chi_{\bmxi}(\mathbf{y})f(\mathbf{y}),\label{eq:fourier-def}
\end{align}
for all $\bmxi\in\ff^N$.
The values $\widehat f(\bmxi)$ are the \emph{Fourier coefficients} of $f$, and the expansion in the character basis is
\begin{align}
f(\mathbf{y})&=\sum_{\bmxi\in\ff^N}\widehat f(\bmxi)\chi_{\bmxi}(\mathbf{y}),
\label{eq:fourier-inverse-prelim}
\end{align}
for all $\yy\in\ff^N$.

The normalized \emph{convolution} $f*g\colon\ff^N\to\cc$ of $f$ and $g$ is defined by
\begin{equation}
\label{eq:conv-def}
(f*g)(\xx)\coloneqq 2^{-N}\sum_{\yy\in\ff^N}f(\yy)g(\xx+\yy),
\end{equation}
for all $\xx\in\ff^N$.
Equation~\eqref{eq:char-mult-argument} yields the convolution theorem,
\begin{equation}
\label{eq:conv-thm}
\widehat{f*g}(\bmxi)=\widehat f(\bmxi) \widehat g(\bmxi),
\end{equation}
for all $\bmxi$.
The function $fg\colon\ff^N\to\cc$ is defined by $(fg)(\xx)\coloneqq f(\xx)g(\xx)$ for all $\xx$, and $fg$ is called the \emph{pointwise product} of $f$ and $g$.

The \textit{Fourier $1$-norm}~\cite{Green_2008_algebra,ODonnell2014Analysis} of $f$ is the sum of the absolute values of its Fourier coefficients,
\begin{equation}
\label{eq:spectral-norm}
\norm{f}_A\coloneqq\sum_{\bmxi\in\ff^N}\abs{\widehat f(\bmxi)}.
\end{equation}
The Fourier $1$-norm has the following properties, that is, homogeneity, the pointwise bound, the triangle inequality, and submultiplicativity, respectively~\cite{ODonnell2014Analysis}.
For all $f$ and $g$, all $c\in\cc$, and all $\xx\in\ff^{N}$,
\begin{align}
\norm{cf}_A&=\abs{c}\norm{f}_A,\label{eq:spectral-homogeneous}\\
\abs{f(\xx)}&\leq\norm{f}_A,\label{eq:pointwise-bound}\\
\norm{f+g}_A&\leq\norm{f}_A+\norm{g}_A,\label{eq:spectral-triangle}\\
\norm{fg}_A&\leq\norm{f}_A\norm{g}_A.\label{eq:spectral-submultiplicative}
\end{align}

\section{Setting}
\label{sec:setting}
In this section, we define the coherent $Z$-rotation error model, transition amplitudes, and the maximum-likelihood Pauli recovery.
We also identify the entanglement infidelity that will be bounded in the remaining sections.

Following the coherent error models studied for the surface code~\cite{Bravyi2018Correcting,VennBeri2020Error,Venn2023Coherent}, we consider $Z$-type errors on the data qubits in this paper.
For angles $\bm{\theta}\coloneqq (\theta_1,\ldots, \theta_n)\in\mathbb{R}^n$ with $|\theta_j|<\pi/4$ for all $j\in[n]$, the data qubits experience the unitary error described by
\begin{equation}
    U_{\bm{\theta}}\coloneqq \bigotimes_{j=1}^n e^{\ii \theta_j Z_j}.
\label{eq: coherent error}
\end{equation}
In addition, we define
\begin{equation}
    \delta\coloneqq \max_{j\in[n]} \abs{\tan{(2\theta_j)}},
\label{eq: delta}
\end{equation}
which is finite under the assumption $|\theta_j|<\pi/4$ for all $j\in[n]$.
The threshold in Theorem~\ref{thm:threshold theorem} is stated in terms of $\delta$.

For $\sss\in\ff^{m_X}$ and $\uu\in\ff^k$, define the \textit{transition amplitude}
\begin{equation}
\label{eq: Asu-def}
A_{\sss,\uu}\coloneqq{}_X\!\bra{\sss,\uu}U_{\bm\theta}\ket{\overline{\bm0}}_X.
\end{equation}
The transition amplitudes describe the action of $U_{\bm\theta}$ on every logical basis state $\ket{\overline{\mathbf p}}_X$, not only $\ket{\overline{\bm0}}_X$.
For any $\mathbf p\in\ff^k$, the commutation of $U_{\bm\theta}$
with $\overline Z(\mathbf p)$ and every $g_Z^{(j)}$ gives
\begin{align}
U_{\bm\theta}\ket{\overline{\mathbf p}}_X
=\sum_{\sss\in\ff^{m_X}}\sum_{\uu\in\ff^k}
A_{\sss,\uu}\ket{\sss,\mathbf p+\uu}_X,
\label{eq:U-X-basis}
\end{align}
where we use Eqs.~\eqref{eq:logical-basis-action},~\eqref{eq:logical-X-basis}, and \eqref{eq: Asu-def}.

Define
\begin{equation}
Q_{\sss,\uu}\coloneqq |A_{\sss,\uu}|^2.
\label{eq: Qsu}
\end{equation}
For input $\ket{\overline{\mathbf p}}_X$, this is the transition probability to $\ket{\sss,\mathbf p+\uu}_X$ when both the syndrome and the logical $X$ basis are measured.
Taking the squared norm of Eq.~\eqref{eq:U-X-basis} at $\mathbf p=\bm0$ gives
\begin{align}
\label{eq:coset-normalization}
\sum_{\sss,\uu}Q_{\sss,\uu}
&=\sum_{\sss,\uu}|A_{\sss,\uu}|^2\nonumber\\
&=\left\|U_{\bm\theta}\ket{\overline{\bm0}}_X\right\|^2=1,
\end{align}
where the second equality uses orthonormality of the states in Eq.~\eqref{eq:logical-X-basis}, and the last uses unitarity of $U_{\bm\theta}$.

After the syndrome $\sss$ is measured, we apply a Pauli recovery.
We first define a Pauli recovery for a general syndrome-dependent decision rule $f\colon\ff^{m_X}\to\ff^k$, and specialize to the maximum-likelihood Pauli recovery below.
For each syndrome $\sss$, define the Pauli recovery 
\begin{equation}
\label{eq:Pauli-recovery-def}
R_\sss^f\coloneqq Z(\ee)    
\end{equation}
for some $\ee\in\ff^n$ satisfying
\begin{equation}
\ee H_X^\top=\sss,\qquad \ee O_X^\top=f(\sss).
\end{equation}
Such a vector $\ee$ exists because the rows of $H_X$ and $O_X$ are linearly independent.
Applying $R_\sss^f$ to $\ket{\sss,\mathbf p}_X$ gives
\begin{align}
R_\sss^f\ket{\sss,\mathbf p}_X
=\ket{\overline{\mathbf p+f(\sss)}}_X
\end{align}
by Eq.~\eqref{eq:Pauli-basis-transition}.
Thus the action on the syndrome subspace is
\begin{equation}
\label{eq:recovery-basis-action}
R_\sss^f\Pi_\sss
=\sum_{\mathbf p\in\ff^k}
\ket{\overline{\mathbf p+f(\sss)}}_X{}_X\!\bra{\sss,\mathbf p},
\end{equation}
which does not depend on the choice of $\ee$ in Eq.~\eqref{eq:Pauli-recovery-def}.
We use this action in Sec.~\ref{sec:fidelity-ml} to compute the Kraus operators.

For each syndrome $\mathbf{s}$, the \textit{maximum-likelihood} (ML) decoder $\hat{\uu}\colon\ff^{m_X}\to\ff^k$ chooses a logical class with the largest transition probability,
\begin{equation}
    \hat{\uu}(\sss)\in\argmax_{\uu\in\ff^k} Q_{\sss,\uu}
\label{eq: hatu recovery ML}
\end{equation}
and applies the Pauli recovery operation
\begin{equation}
    R_\sss^{\mathrm{ML}}\coloneqq R_\sss^{\hat{\uu}}.
\label{eq: optimal Pauli recovery}
\end{equation}
Here the rotation angles $\bm{\theta}$ are used to determine the transition probabilities $Q_{\sss,\uu}$ of Eq.~\eqref{eq: Qsu}, while the recovery remains the syndrome-dependent Pauli operator of Eq.~\eqref{eq: optimal Pauli recovery}.
We do not claim optimality over arbitrary recovery operations, such as a coherent counter-rotation constructed directly from the angles.
Our motivation is to ask whether syndrome measurement followed by Pauli recovery is sufficient to suppress coherent errors.
As shown in Proposition~\ref{prop:Fe-general}, the ML recovery maximizes the entanglement fidelity among all syndrome-dependent Pauli recoveries of this form, and therefore provides the benchmark for analyses of other decoders such as minimum-weight decoders~\cite{Dennis_2002,Higgott2022PyMatching,Gottesman2014constant,Takada2026Doubly}.

Let $\Lambda^{\mathrm{ML}}\colon \mathcal{L}\qty(\qty(\cc^2)^{\otimes k})\to \mathcal{L}\qty(\qty(\cc^2)^{\otimes k})$ denote the logical channel obtained by encoding into the codespace, applying the coherent $Z$-rotation error in Eq.~\eqref{eq: coherent error}, measuring the syndrome, and performing the recovery in Eq.~\eqref{eq: optimal Pauli recovery}, where $\mathcal Q$ is identified with $(\cc^2)^{\otimes k}$ as a $2^k$-dimensional Hilbert space.
The Kraus operators and the entanglement fidelity of the logical channel $\Lambda^{\mathrm{ML}}$ are derived in Sec.~\ref{sec:fidelity-ml}.
In particular,
\begin{equation}
    F^{\mathrm{ML}}_{e}\coloneqq F_e(\Lambda^{\mathrm{ML}})=\sum_{\sss\in\ff^{m_X}}\max_{\uu\in\ff^k}Q_{\sss,\uu}.
\label{eq: fml qsu}
\end{equation}
The infidelity $1-F_e^{\mathrm{ML}}$ is the quantity of interest, and Theorem~\ref{thm:threshold theorem} below shows that it is exponentially small in the code distance up to a prefactor linear in the number of data qubits.

\section{Main Results}
\label{sec:main-results}
The following theorem is the main result of our paper.
The theorem shows that the infidelity of each code in the family is suppressed exponentially in the code distance, up to a prefactor linear in the number of data qubits.
Under a growth condition on the code distance, the same bound gives a positive threshold in the code-capacity setting, with the coherent $Z$-rotation errors of Eq.~\eqref{eq: coherent error}, ideal syndrome measurement, and the maximum-likelihood Pauli recovery of Eq.~\eqref{eq: optimal Pauli recovery}.

\begin{theorem}
\label{thm:threshold theorem}
    Let $\{\mathcal{Q}_i\}_{i\geq 1}$ be a family of CSS $(r,c)$-QLDPC codes with code parameters $[[n_i,k_i,d_i]]$, where $n_i\to\infty$ as $i\rightarrow \infty$, and let $d_{Z,i}$ denote the minimum weight of a nontrivial logical $Z$ operator given in Eq.~\eqref{eq:dz-def}.
    Assume that there exists a constant $k_0$ satisfying
    \begin{align}
    \label{eq:k-assumption}
        1\leq k_i\leq k_0,
    \end{align}
    for all $i$.
    Define $\Delta_0=c(r-1)$, where we assume $r\geq 2$ and $c\geq 1$.
    Suppose that for each $i$, the data qubits of $\mathcal{Q}_i$ experience the coherent $Z$-rotation error of Eq.~\eqref{eq: coherent error} with angles $\bm \theta_i\in(-\pi/4,\pi/4)^{n_i}$, followed by the ideal syndrome measurement and the maximum-likelihood recovery of Eq.~\eqref{eq: optimal Pauli recovery}. 
    Let $F^{\mathrm{ML}}_{e,i}$ be the entanglement fidelity of the resulting logical channel $\Lambda^{\mathrm{ML}}_i$, and suppose that a constant $\delta_*>0$ satisfies
\begin{equation}
\label{eq:delta-theorem}
    \max_{j\in[n_i]}
    \abs{\tan(2\theta_{i,j})}
    \leq\delta_*
\end{equation}
for all $i$.
    Define
    \begin{equation}
    \label{eq:delta-th-def}
    \delta_{\mathrm{th}}\coloneqq \frac{1}{e[(e+1)\Delta_0+1]}.
    \end{equation}
   If $0<\delta_*<\delta_{\mathrm{th}}$, one may take 
   \begin{align}
        b&\coloneqq\frac{1}{2}\log\qty(\frac{\delta_{\mathrm{th}}}{\delta_*}),\label{eq:constant-b}\\
        C&\coloneqq \qty(2^{k_0}-1)^2\label{eq:constant-C},
    \end{align}
    and for all $i$,
    \begin{equation}
    \label{eq:main-bound}
        1-F^{\mathrm{ML}}_{e,i}
        \leq Cn_i e^{-bd_{Z,i}}
        \leq Cn_i e^{-bd_i}.
    \end{equation}
    In addition, if $d_i/\log n_i\rightarrow \infty$, then
    $1-F^{\mathrm{ML}}_{e,i}\to 0$ as $i\to\infty$.
\end{theorem}

The rotated planar surface codes~\cite{Bravyi1998surface,Bombin2007Optimal,Tomita2014LowDistance} with distance $d$ are $[[d^2,1,d]]$ codes whose $X$-type and $Z$-type parity-check matrices each have row weight at most $4$ and column weight at most $2$, so this family is a CSS $(4,2)$-QLDPC family in our definition with $k_i=1$ and $d_i=\sqrt{n_i}$ for every $i$, while $n_i\to\infty$, and it satisfies the assumptions of Theorem~\ref{thm:threshold theorem} with $k_0=1$ and the condition $d_i/\log n_i\to\infty$.
Then $\Delta_0=c(r-1)=6$, and Theorem~\ref{thm:threshold theorem} gives a lower bound $\delta_{\mathrm{th}}$ of Eq.~\eqref{eq:delta-th-def} on the threshold,
\begin{equation}
\delta_{\mathrm{th}}=\frac{1}{e(6e+7)}\approx1.6\times10^{-2}.
\end{equation}
For a fixed code, the condition $\max_j\abs{\tan{(2\theta_j)}}<\delta_{\mathrm{th}}$ is equivalent to
\begin{equation}
\label{eq:theta-th-surface}
\max_{j}|\theta_{j}|<\frac{1}{2}\arctan\qty(\delta_\mathrm{th})\approx2.5\times10^{-3}\pi.
\end{equation}
This value is small compared to the thresholds for coherent errors on the surface code estimated numerically in Refs.~\cite{Bravyi2018Correcting,Venn2023Coherent}.
Our contribution is not in obtaining a high threshold value, but in presenting a rigorous proof that a positive threshold exists for the coherent error model.

Beyond the surface code, Theorem~\ref{thm:threshold theorem} applies to other topological codes~\cite{Fujii2015Topological,Dennis_2002,Kitaev2003anyons,bombin2013introductiontopologicalquantumcodes} with a fixed topology and $d_i/\log n_i\to\infty$, such as the toric code~\cite{Dennis_2002,Kitaev2003anyons,Kitaev1997Quantum}, the color code~\cite{Bombin2006Topological}, and toric codes on higher-dimensional lattices~\cite{Dennis_2002}, since each of these families is a CSS QLDPC family with a bounded number of logical qubits.
For the toric code, linearly dependent rows of the parity-check matrices may be discarded as described in Sec.~\ref{subsec:prelim-css} without changing the code or increasing the row or column weights.

\begin{figure}[t]
\centering
\includegraphics[width=\columnwidth]{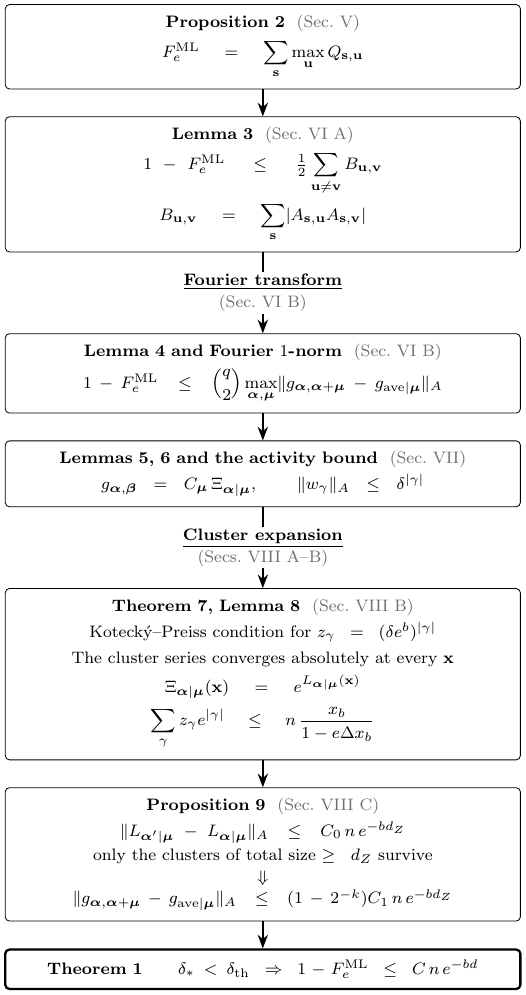}
\caption{Roadmap of the proof of Theorem~\ref{thm:threshold theorem}, showing the statement, key identity or bound, and sections for each step.
The boxes indicate the sections in which the corresponding statements are established.}
\label{fig:roadmap}
\end{figure}

We next discuss the assumption $k_i\leq k_0$.
Before using this assumption, the proof gives the bound
\begin{equation}
\label{eq:codewise-bound}
    1-F^{\mathrm{ML}}_{e,i}
    \leq (2^{k_i}-1)^2n_i e^{-bd_{Z,i}}.
\end{equation}
The assumption is used to make the prefactor uniform in the family.
Note that, more generally, Eq.~\eqref{eq:codewise-bound} implies $1-F^{\mathrm{ML}}_{e,i}\to0$ whenever
\begin{equation}
\label{eq:general-assump}
    bd_{Z,i}-\log n_i-2\log\qty(2^{k_i}-1)\longrightarrow\infty.
\end{equation}
Even without the assumption $k_i\leq k_0$, this condition holds for code families with distance growing linearly with $n_i$, such as asymptotically good CSS QLDPC codes~\cite{PanteleevKalachev2022Asymptotically,Dinur2023Good,Leverrier2022Quantum,Christandl2025Fault}, if $\delta_*$ is sufficiently small.

In Refs.~\cite{Christandl2025Fault,Aharonov2008Fault,Kitaev1997Quantum,Iverson2020Coherence}, the strength of the noise is measured by the deviation of the noisy channel from the ideal channel in the diamond norm.
Unlike the entanglement fidelity, which evaluates the channel on the maximally entangled input, the diamond distance measures the worst-case deviation from the identity channel over all inputs.
In terms of the diamond distance, Eq.~\eqref{eq:diamond-Fe} turns Eq.~\eqref{eq:main-bound} into
\begin{equation}
\label{eq:main-diamond}
    \frac{1}{2}\bigl\|\Lambda^{\mathrm{ML}}_i-\id\bigr\|_\diamond\leq 2^{k_0}\sqrt{Cn_i}e^{-bd_i/2}
\end{equation}
for the logical channel $\Lambda^{\mathrm{ML}}_i$ of $\mathcal{Q}_i$.
The right-hand side of Eq.~\eqref{eq:main-diamond} also tends to zero if $d_i/\log n_i\to\infty$.

The remainder of the paper focuses on proving Theorem~\ref{thm:threshold theorem}.
Figure~\ref{fig:roadmap} summarizes the proof.

\section{Entanglement fidelity under the maximum-likelihood Pauli recovery}
\label{sec:fidelity-ml}
Our Theorem~\ref{thm:threshold theorem} bounds the infidelity $1-F_e^{\mathrm{ML}}$.
In this section, we use the bases of the syndrome subspaces introduced in Sec.~\ref{subsec:prelim-syndrome-bases} and derive Eq.~\eqref{eq: fml qsu} from the action of the Kraus operators.

For a syndrome-dependent decision rule $f\colon\ff^{m_X}\to\ff^k$, let $\Lambda^f$ denote the logical channel obtained by using the recovery $R_\sss^f$ of Eq.~\eqref{eq:Pauli-recovery-def}.
Here, $f(\sss)$ specifies the logical class of the Pauli recovery applied after the syndrome $\sss$ is measured.
The syndrome measurement on $\mathcal{Q}$ with outcome $\sss$ followed by the recovery $R_\sss^f$ gives the operator
\begin{equation}
K_\sss^f\coloneqq R_\sss^f\Pi_\sss U_{\bm\theta}\Pi_\mathcal{Q}.
\label{eq:Ks-def}
\end{equation}
Summing $(K_\sss^f)^\dagger K_\sss^f$ over the syndromes gives
\begin{align}
\label{eq: complete condition}
\sum_{\sss\in\ff^{m_X}}(K_\sss^f)^\dagger K_\sss^f
&=\Pi_\mathcal Q U_{\bm\theta}^\dagger
\left(\sum_{\sss\in\ff^{m_X}}\Pi_\sss\right)
U_{\bm\theta}\Pi_\mathcal Q\nonumber\\
&=\Pi_\mathcal Q,
\end{align}
where the second equality uses Eq.~\eqref{eq:syndrome-completeness} together with the commutation of $U_{\bm \theta}$ with every $g_Z^{(i)}$.
Thus $\{K_\sss^f\}_{\sss\in\ff^{m_X}}$ are the Kraus operators for $\Lambda^f$ on $\mathcal Q$.

Applying the syndrome projector and the recovery to Eq.~\eqref{eq:U-X-basis} gives
\begin{align}
K_\sss^f\ket{\overline{\mathbf p}}_X
&=R_\sss^f\Pi_\sss\sum_{\sss',\uu}A_{\sss',\uu}\ket{\sss',\mathbf p+\uu}_X\nonumber\\
&=\sum_{\uu\in\ff^k}A_{\sss,\uu}\ket{\overline{\mathbf p+\uu+f(\sss)}}_X\nonumber\\
&=\sum_{\uu\in\ff^k}A_{\sss,\uu}\overline Z\bigl(f(\sss)+\uu\bigr)\ket{\overline{\mathbf p}}_X,
\label{eq:Ks-X-basis-action}
\end{align}
where the sum in the first line is over $\sss'\in\ff^{m_X}$ and $\uu\in\ff^k$, the second equality uses Eq.~\eqref{eq:recovery-basis-action}, and the last equality uses Eqs.~\eqref{eq:logical-Z-codeword-action} and~\eqref{eq:logical-X-codeword-basis}.
This gives the Kraus operator in the form
\begin{equation}
\label{eq:Ks-form}
K_\sss^f=\left(\sum_{\uu\in\ff^k}A_{\sss,\uu}\overline Z\bigl(f(\sss)+\uu\bigr)\right)\Pi_\mathcal Q.
\end{equation}
Taking the trace in the logical $X$ basis gives
\begin{align}
\label{eq:Ks-trace}
\Tr K_\sss^f
&=\sum_{\mathbf p\in\ff^k}{}_X\!\bra{\overline{\mathbf p}}K_\sss^f\ket{\overline{\mathbf p}}_X\nonumber\\
&=qA_{\sss,f(\sss)},
\end{align}
where $q=2^k$, and the last equality uses Eq.~\eqref{eq:Ks-X-basis-action}.

The following proposition gives the entanglement fidelity of $\Lambda^f$ for every decision rule $f$ and shows that the ML recovery maximizes it among these rules.
\begin{proposition}
\label{prop:Fe-general}
For every syndrome-dependent decision rule $f\colon\ff^{m_X}\to\ff^k$,
\begin{equation}
\label{eq:Fe-decision-rule}
F_e(\Lambda^f)=\sum_{\sss\in\ff^{m_X}}Q_{\sss,f(\sss)}.
\end{equation}
In particular, the entanglement fidelity of $\Lambda^{\mathrm{ML}}$ is given by Eq.~\eqref{eq: fml qsu}, that is,
\begin{equation}
\label{eq:Qsu-ML-second}
F_e^{\mathrm{ML}}=\sum_{\sss\in\ff^{m_X}}\max_{\uu\in\ff^k}Q_{\sss,\uu},
\end{equation}
and
\begin{equation}
\label{eq:Fe-inequality}
F_e(\Lambda^{\mathrm{ML}})\geq F_e(\Lambda^f)
\end{equation}
for every syndrome-dependent decision rule $f$.
\end{proposition}
\begin{proof}
Equations~\eqref{eq:Fe-kraus} and~\eqref{eq:Ks-trace} give
\begin{equation}
F_e(\Lambda^f)=q^{-2}\sum_{\sss\in\ff^{m_X}}|\Tr K_\sss^f|^2
=\sum_{\sss\in\ff^{m_X}}Q_{\sss,f(\sss)},
\end{equation}
where the last equality also uses Eq.~\eqref{eq: Qsu}.
This is Eq.~\eqref{eq:Fe-decision-rule}.
Taking $f=\widehat{\uu}$ gives
\begin{equation*}
F_e^{\mathrm{ML}}=\sum_\sss Q_{\sss,\widehat{\uu}(\sss)}
=\sum_\sss\max_\uu Q_{\sss,\uu},
\end{equation*}
where the last equality uses Eq.~\eqref{eq: hatu recovery ML}.
This is Eq.~\eqref{eq:Qsu-ML-second}, and
\begin{equation*}
F_e(\Lambda^{\mathrm{ML}})-F_e(\Lambda^f)
=\sum_\sss\left[\max_\uu Q_{\sss,\uu}-Q_{\sss,f(\sss)}\right]\geq0,
\end{equation*}
which proves Eq.~\eqref{eq:Fe-inequality}.
\end{proof}

Proposition~\ref{prop:Fe-general} is the starting point of the threshold proof. 
The fidelity is the sum, over the syndromes, of the transition probability of the logical class selected by the recovery.

\section{Bounding entanglement infidelity by Fourier \texorpdfstring{$1$}{1}-norm}
\label{sec:fourier-rep}
In this section, we bound the infidelity $1-F_e^{\mathrm{ML}}$ of Eq.~\eqref{eq:Qsu-ML-second} in two steps.
In Sec.~\ref{subsec:pairwise-bound} we bound the infidelity by a sum over pairs of distinct logical classes.
In Sec.~\ref{subsec:syndrome-fourier} we express the amplitudes $A_{\sss,\uu}$ in the computational basis and use the restriction $\uu\neq\vv$ before taking absolute values to obtain the Fourier $1$-norm bound of Eq.~\eqref{eq:master-reduction}.

\subsection{Pairwise bound on the infidelity}
\label{subsec:pairwise-bound}
We replace the maximum in Eq.~\eqref{eq: fml qsu} by a sum of pairwise quantities indexed by logical classes.
For $\uu,\vv\in\ff^k$, define
\begin{equation}
    B_{\uu,\vv}\coloneqq\sum_{\sss\in\ff^{m_X}}\sqrt{Q_{\sss,\uu}Q_{\sss,\vv}}=\sum_{\sss\in\ff^{m_X}}\abs{A_{\sss,\uu}A_{\sss,\vv}},
\label{eq:B-def}
\end{equation}
where we use Eq.~\eqref{eq: Qsu}.
We call $B_{\uu,\vv}$ the \emph{pairwise overlap} of the logical classes $\uu$ and $\vv$.
The quantity measures how strongly the same syndromes can simultaneously support two logical classes.
The following lemma reduces the ML infidelity to the quantities
$B_{\uu,\vv}$.

\begin{lemma}
\label{lem:pairwise-bound}
\begin{equation}
\label{eq:pairwise-bound}
1-F_e^{\mathrm{ML}}\leq\frac12
\sum_{\substack{\uu,\vv\in\ff^k\\\uu\neq\vv}}B_{\uu,\vv}.
\end{equation}
\end{lemma}
\begin{proof}
By Eq.~\eqref{eq:coset-normalization} and Proposition~\ref{prop:Fe-general},
\begin{align}
1-F_e^{\mathrm{ML}}
&=\sum_\sss\left(\sum_\uu Q_{\sss,\uu}-\max_\uu Q_{\sss,\uu}\right)\nonumber\\
&=\sum_\sss\sum_{\uu\neq\widehat{\uu}(\sss)}Q_{\sss,\uu},
\end{align}
where $\widehat{\uu}(\sss)\in\argmax_\uu Q_{\sss,\uu}$.
For fixed $\sss$, we bound the inner sum by a sum over all different pairs,
\begin{align}\nonumber
&\sum_{\uu\neq\widehat{\uu}(\sss)}Q_{\sss,\uu}\\
&\leq\sum_{\uu\neq\widehat{\uu}(\sss)}\sqrt{Q_{\sss,\uu}Q_{\sss,\widehat{\uu}(\sss)}}\nonumber\\
&\leq\frac12\sum_{\substack{\uu,\vv\in\ff^k\\\uu\neq\vv}}
\sqrt{Q_{\sss,\uu}Q_{\sss,\vv}},
\end{align}
where the first inequality uses $Q_{\sss,\uu}\leq Q_{\sss,\widehat{\uu}(\sss)}$.
The last inequality follows because the ordered pairs $(\uu,\widehat{\uu}(\sss))$ and $(\widehat{\uu}(\sss),\uu)$, over $\uu\neq\widehat{\uu}(\sss)$, are all distinct.
Summing over $\sss$ gives
\begin{align*}
1-F_e^{\mathrm{ML}}
&\leq\frac12\sum_{\substack{\uu,\vv\in\ff^k\\\uu\neq\vv}}
\sum_{\sss\in\ff^{m_X}}\sqrt{Q_{\sss,\uu}Q_{\sss,\vv}}\\
&=\frac12\sum_{\substack{\uu,\vv\in\ff^k\\\uu\neq\vv}}B_{\uu,\vv},
\end{align*}
where the last equality is Eq.~\eqref{eq:B-def}.
This is Eq.~\eqref{eq:pairwise-bound}.
\end{proof}

\noindent\emph{Remark.}
This remark expresses $A_{\sss,\uu}$ in terms of the amplitudes of error patterns and shows that replacing each amplitude by its absolute value, which discards the cancellations among the amplitudes, may not give the bound of Theorem~\ref{thm:threshold theorem}.
The dependence of $A_{\sss,\uu}$ on the rotation angles is given by
\begin{equation}
\label{eq:A-coset-sum}
A_{\sss,\uu}
=\sum_{\substack{\ee\in\ff^n\\
\ee H_X^\top=\sss,\ \ee O_X^\top=\uu}}a_\ee,
\end{equation}
with
\begin{equation}
\label{eq:ae-def}
a_\ee\coloneqq\prod_{j=1}^n
(\cos\theta_j)^{1-e_j}(\ii\sin\theta_j)^{e_j},
\end{equation}
where $a_\ee$ is the amplitude of the error pattern $Z(\ee)$ of $U_{\bm\theta}$.
This follows by substituting $U_{\bm\theta}=\sum_{\ee\in\ff^n}a_\ee Z(\ee)$ into Eq.~\eqref{eq: Asu-def} and using Eq.~\eqref{eq:Pauli-basis-transition} and the orthonormality of the states in Eq.~\eqref{eq:logical-X-basis}.
Applying the triangle inequality to the sums over error patterns in Eq.~\eqref{eq:A-coset-sum}, which replaces each $a_\ee$ by $\abs{a_\ee}$, gives
\begin{align}
\label{eq:pairwise-pattern-bound}
\frac12\sum_{\substack{\uu,\vv\in\ff^k\\\uu\neq\vv}}B_{\uu,\vv}
&\leq\frac12\sum_{\substack{\ee,\ee'\in\ff^n\\
(\ee+\ee')H_X^\top=\bm0\\
(\ee+\ee')O_X^\top\neq\bm0}}\abs{a_\ee a_{\ee'}}\nonumber\\
&=\frac12\sum_{\substack{\ccc\in\ff^n\\
\ccc H_X^\top=\bm0\\
\ccc O_X^\top\neq\bm0}}
\prod_{j:c_j=1}\abs{\sin(2\theta_j)},
\end{align}
where the sum after the inequality runs over the pairs of patterns with the same syndrome and different logical classes in Eq.~\eqref{eq:B-def}, and the equality uses $\ccc=\ee+\ee'$ and
\begin{equation}
\sum_{\ee\in\ff^n}\abs{a_\ee a_{\ee+\ccc}}
=\prod_{j:c_j=1}\abs{\sin(2\theta_j)}.
\end{equation}
Indeed, summing over $e_j$ gives $\cos^2\theta_j+\sin^2\theta_j=1$ when $c_j=0$, and $2\abs{\cos\theta_j\sin\theta_j}=\abs{\sin(2\theta_j)}$ when $c_j=1$.
Every vector retained in Eq.~\eqref{eq:pairwise-pattern-bound} has weight at least $d_Z$ by Eq.~\eqref{eq:dz-def}.
However, the addition of $Z$-type stabilizers preserves its syndrome and logical class.
Let $\ccc_0$ represent a nontrivial logical $Z$ operator with $\wt(\ccc_0)=d_Z$, and set $\mathbf z_i=\mathbf h_Z^{(i)}$ for $i\in[m_Z]$.
These vectors are linearly independent by the full row rank of $H_Z$ and have weights at most $r$ by the property of the $(r,c)$-QLDPC code.

Suppose that all rotation angles equal a fixed nonzero $\theta$, and set $t=\abs{\sin(2\theta)}\in(0,1)$.
Summing over the addition of all subsets of these stabilizers gives a lower bound on the right-hand side of
Eq.~\eqref{eq:pairwise-pattern-bound},
\begin{align}
\label{eq:exponential growth}
\frac12\sum_{\substack{\ccc\in\ff^n\\
\ccc H_X^\top=\bm0\\
\ccc O_X^\top\neq\bm0}}
t^{\wt(\mathbf{c})}&\geq \frac12\sum_{S\subseteq[m_Z]}
t^{\wt(\ccc_0+\sum_{i\in S}\mathbf z_i)}\nonumber\\
&\geq\frac12\sum_{S\subseteq[m_Z]}
t^{d_Z+\sum_{i\in S}\wt(\mathbf z_i)}\nonumber\\
&=
\frac{t^{d_Z}}2
\prod_{i=1}^{m_Z}
\left(1+t^{\wt(\mathbf z_i)}\right)\nonumber\\
&\geq
\frac{t^{d_Z}}2(1+t^r)^{m_Z},
\end{align}
where $S$ gives the indices of the added stabilizers.
Different subsets give distinct vectors because the stabilizer vectors are linearly independent.
The last expression grows exponentially in $n$ if $d_Z/n\to0$.
Thus the bounded weights of the stabilizer generators give exponential growth on the right-hand side of Eq.~\eqref{eq:exponential growth} when the cancellations among amplitudes are discarded.
In contrast, in Sec.~\ref{sec:cluster}, we use the sparsity of the QLDPC code to obtain the bound of Theorem~\ref{thm:threshold theorem}.

\subsection{Fourier \texorpdfstring{$1$}{1}-norm bound}
\label{subsec:syndrome-fourier}

In this subsection, we express the transition amplitudes $A_{\sss,\uu}$ through the phases of the computational-basis states under $U_{\bm\theta}$ and derive the Fourier $1$-norm bound of Eq.~\eqref{eq:master-reduction}.
The unitary $U_{\bm\theta}$ acts diagonally on the computational basis.
For $\bmalpha\in\ff^k$ and $\yy\in\ff^{m_X}$, define the phase of the computational-basis state $\ket{\bmalpha O_X+\yy H_X}$ under $U_{\bm\theta}$ by
\begin{align}
\label{eq:falpha-def}
f_{\bmalpha}(\yy)
&\coloneqq
\bra{\bmalpha O_X+\yy H_X}U_{\bm\theta}
\ket{\bmalpha O_X+\yy H_X}\nonumber\\
&=\prod_{j=1}^{n}\exp[\ii\theta_j(-1)^{(\bmalpha O_X)_j+(\yy H_X)_j}].
\end{align}
All factors in Eq.~\eqref{eq:falpha-def} have absolute value one, so $\abs{f_{\bmalpha}(\yy)}=1$ for all $\bmalpha\in\ff^k$ and $\yy\in\ff^{m_X}$.
Expanding the logical $X$-basis states in the computational basis gives
\begin{align}
\label{eq:syndrome-X-computational}
&\ket{\sss,\uu}_X\nonumber\\
&=2^{-(k+m_X)/2}
\sum_{\bmalpha\in\ff^k}\sum_{\yy\in\ff^{m_X}}
(-1)^{\bmalpha\cdot\uu+\sss\cdot\yy}
\ket{\bmalpha O_X+\yy H_X},
\end{align}
where we use Eqs.~\eqref{eq:syndrome-logical-basis} and~\eqref{eq:logical-X-basis}.
Substituting Eq.~\eqref{eq:syndrome-X-computational} into Eq.~\eqref{eq: Asu-def} gives
\begin{equation}
\label{eq:A-computational}
A_{\sss,\uu}
=2^{-k-m_X}
\sum_{\bmalpha\in\ff^k}\sum_{\yy\in\ff^{m_X}}
(-1)^{\bmalpha\cdot\uu+\sss\cdot\yy}f_{\bmalpha}(\yy),
\end{equation}
where $\ket{\overline{\bm0}}_X=\ket{\bm0,\bm0}_X$.
To express the products $A_{\sss,\uu}A_{\sss,\vv}$, we introduce the convolution of the functions $f_{\bmalpha}$ and $f_{\bmbeta}$.
For $\bmalpha,\bmbeta\in\ff^k$, define $g_{\bmalpha,\bmbeta}\colon\ff^{m_X}\to\cc$ by
\begin{equation}
\label{eq:g-def}
g_{\bmalpha,\bmbeta}(\xx)
\coloneqq(f_{\bmalpha}*f_{\bmbeta})(\xx)
=2^{-m_X}\sum_{\yy\in\ff^{m_X}}f_{\bmalpha}(\yy)f_{\bmbeta}(\xx+\yy).
\end{equation}
Substituting Eq.~\eqref{eq:A-computational} for both factors of $A_{\sss,\uu}A_{\sss,\vv}$ and setting $\bmbeta=\bmalpha+\bmmu$ and $\yy'=\yy+\xx$ gives
\begin{align}
\label{eq:AA-double-sum}
&A_{\sss,\uu}A_{\sss,\vv}=2^{-2k-2m_X}\nonumber\\
&\quad\times\sum_{\substack{\bmalpha,\bmmu\in\ff^k\\\xx,\yy\in\ff^{m_X}}}
\bigl[(-1)^{\bmalpha\cdot(\uu+\vv)+\bmmu\cdot\vv+\sss\cdot\xx}f_{\bmalpha}(\yy)f_{\bmalpha+\bmmu}(\yy+\xx)\bigr]\nonumber\\
&=2^{-2k-m_X}\sum_{\substack{\bmalpha,\bmmu\in\ff^k\\\xx\in\ff^{m_X}}}
\bigl[(-1)^{\bmalpha\cdot(\uu+\vv)+\bmmu\cdot\vv+\sss\cdot\xx}g_{\bmalpha,\bmalpha+\bmmu}(\xx)\bigr],
\end{align}
where $(\bmalpha,\bmbeta,\yy,\yy')\mapsto(\bmalpha,\bmmu,\yy,\xx)$ is a bijection, the signs satisfy $\sss\cdot\yy+\sss\cdot\yy'=\sss\cdot\xx$, and the last equality uses Eq.~\eqref{eq:g-def}.

To use the restriction $\uu\neq \vv$ of Lemma~\ref{lem:pairwise-bound}, we subtract from $g_{\bmalpha,\bmalpha+\bmmu}$ in Eq.~\eqref{eq:AA-double-sum} its average over $\bmalpha$.
For fixed $\bmmu\in\ff^k$, define the average $g_{\mathrm{ave}|\bmmu}\colon\ff^{m_X}\to\cc$ over $\bmalpha$ by
\begin{equation}
\label{eq:g-average-def}
g_{\mathrm{ave}|\bmmu}
\coloneqq2^{-k}\sum_{\bmalpha'\in\ff^k}
g_{\bmalpha',\bmalpha'+\bmmu}.
\end{equation}
The following lemma shows that the subtraction leaves the products with $\uu\neq\vv$ unchanged and gives the resulting Fourier $1$-norm bound.
\begin{lemma}
\label{lem:average-reduction}
For all $\uu,\vv\in\ff^k$ and all $\sss\in\ff^{m_X}$,
\begin{align}
&A_{\sss,\uu}A_{\sss,\vv} \bm1[\uu\neq\vv]=2^{-2k-m_X}\nonumber\\
&\quad\times\sum_{\substack{\bmalpha,\bmmu\in\ff^k\\\xx\in\ff^{m_X}}}
(-1)^{\bmalpha\cdot(\uu+\vv)+\bmmu\cdot\vv+\sss\cdot\xx}
\bigl[g_{\bmalpha,\bmalpha+\bmmu}(\xx)-g_{\mathrm{ave}|\bmmu}(\xx)\bigr],
\label{eq:AA-subtracted}
\end{align}
and consequently, for $\uu\neq\vv$,
\begin{equation}
\label{eq:B-reduced}
B_{\uu,\vv}
\leq\max_{\bmalpha,\bmmu\in\ff^k}
\norm{g_{\bmalpha,\bmalpha+\bmmu}-g_{\mathrm{ave}|\bmmu}}_A.
\end{equation}
\end{lemma}
\begin{proof}
We first prove Eq.~\eqref{eq:AA-subtracted}.
For each $\bmmu\in\ff^k$ and $\xx\in\ff^{m_X}$, $g_{\mathrm{ave}|\bmmu}(\xx)$ does not depend on $\bmalpha$, so
\begin{align}
&\sum_{\bmalpha\in\ff^k}(-1)^{\bmalpha\cdot(\uu+\vv)}g_{\mathrm{ave}|\bmmu}(\xx)\nonumber\\
&=g_{\mathrm{ave}|\bmmu}(\xx)\sum_{\bmalpha\in\ff^k}(-1)^{\bmalpha\cdot(\uu+\vv)}=0,
\end{align}
where the last equality uses Eq.~\eqref{eq:parity-projector}, since $\uu+\vv\neq\bm0$.
Subtracting $g_{\mathrm{ave}|\bmmu}(\xx)$ from $g_{\bmalpha,\bmalpha+\bmmu}(\xx)$ in Eq.~\eqref{eq:AA-double-sum} therefore leaves $A_{\sss,\uu}A_{\sss,\vv}$ unchanged and proves Eq.~\eqref{eq:AA-subtracted} for $\uu\neq\vv$.
For $\uu=\vv$, the factor $(-1)^{\bmalpha\cdot(\uu+\vv)}$ in Eq.~\eqref{eq:AA-subtracted} equals one, and Eq.~\eqref{eq:g-average-def} gives
\begin{align}
&\sum_{\bmalpha\in\ff^k}
\qty[g_{\bmalpha,\bmalpha+\bmmu}(\xx)-g_{\mathrm{ave}|\bmmu}(\xx)]\nonumber\\
&=\sum_{\bmalpha\in\ff^k}g_{\bmalpha,\bmalpha+\bmmu}(\xx)
-2^kg_{\mathrm{ave}|\bmmu}(\xx)=0.
\end{align}
Equation~\eqref{eq:AA-subtracted} therefore holds for $\uu=\vv$, since the right-hand side is zero and $\bm1[\uu\neq\vv]=0$.

To prove Eq.~\eqref{eq:B-reduced} for $\uu\neq\vv$, we take the absolute value of Eq.~\eqref{eq:AA-subtracted} and apply the triangle inequality to the sum over $\bmalpha,\bmmu$.
The sum over $\xx$ remains inside the absolute value.
Summing over $\sss$ and using Eq.~\eqref{eq:B-def} gives
\begin{align}
\label{eq:B-fourier-bound}
&B_{\uu,\vv}\leq2^{-2k}\nonumber\\
&\enspace\times\sum_{\substack{\bmalpha,\bmmu\in\ff^k\\\sss\in\ff^{m_X}}}
\biggl|2^{-m_X}\sum_{\xx\in\ff^{m_X}}(-1)^{\sss\cdot\xx}
\bigl[g_{\bmalpha,\bmalpha+\bmmu}(\xx)-g_{\mathrm{ave}|\bmmu}(\xx)\bigr]\biggr|\nonumber\\
&=2^{-2k}\sum_{\bmalpha,\bmmu\in\ff^k}
\norm{g_{\bmalpha,\bmalpha+\bmmu}-g_{\mathrm{ave}|\bmmu}}_A\nonumber\\
&\leq\max_{\bmalpha,\bmmu\in\ff^k}
\norm{g_{\bmalpha,\bmalpha+\bmmu}-g_{\mathrm{ave}|\bmmu}}_A,
\end{align}
where the equality uses Eqs.~\eqref{eq:fourier-def} and~\eqref{eq:spectral-norm}, and the last inequality uses that $2^{-2k}\sum_{\bmalpha,\bmmu}$ is an average over all pairs.
\end{proof}

The products $A_{\sss,\uu}A_{\sss,\vv}$ are sums over pairs of errors with the same syndrome.
For errors $\ee$ and $\ee'$ with syndrome $\sss$ and logical classes $\uu$ and $\vv$, respectively,
\begin{equation*}
(\ee+\ee')H_X^\top=\bm0,
\qquad
(\ee+\ee')O_X^\top=\uu+\vv.
\end{equation*}
Hence $Z(\ee+\ee')$ is a $Z$-type stabilizer when $\uu=\vv$, while it is a nontrivial logical $Z$ operator of weight at least $d_Z$ when $\uu\neq\vv$.
Thus the subtraction in Eq.~\eqref{eq:AA-subtracted} removes the pairs that differ by a stabilizer and retains the pairs that differ by a nontrivial logical $Z$ operator.

The right-hand side of Eq.~\eqref{eq:B-reduced} does not depend on the pair $(\uu,\vv)$, so the same bound applies to all $q(q-1)$ ordered pairs with $\uu\neq\vv$ in Lemma~\ref{lem:pairwise-bound}.
Combining Lemmas~\ref{lem:pairwise-bound} and~\ref{lem:average-reduction}, with $q(q-1)/2=\binom{q}{2}$, gives
\begin{equation}
\label{eq:master-reduction}
1-F_e^{\mathrm{ML}}\leq\binom{q}{2}
\max_{\bmalpha,\bmmu\in\ff^k}
\norm{g_{\bmalpha,\bmalpha+\bmmu}-g_{\mathrm{ave}|\bmmu}}_A.
\end{equation}

For use in Sec.~\ref{sec:cluster}, we next prove $\norm{g_{\bmalpha,\bmbeta}}_A\leq1$.
For every $\bmalpha\in\ff^k$, 
\begin{align}
&\sum_{\sss\in\ff^{m_X}}\abs{\widehat{f_{\bmalpha}}(\sss)}^2\nonumber\\
&=\sum_{\sss\in\ff^{m_X}}
\abs{2^{-m_X}\sum_{\yy\in\ff^{m_X}}(-1)^{\sss\cdot\yy}f_{\bmalpha}(\yy)}^2\nonumber\\
&=2^{-2m_X}\sum_{\sss,\yy,\yy'\in\ff^{m_X}}
(-1)^{\sss\cdot(\yy+\yy')}
f_{\bmalpha}(\yy)\overline{f_{\bmalpha}(\yy')}\nonumber\\
&=2^{-m_X}\sum_{\yy\in\ff^{m_X}}\abs{f_{\bmalpha}(\yy)}^2=1,
\label{eq:f-l2}
\end{align}
where the first equality uses Eq.~\eqref{eq:fourier-def}, the third uses Eq.~\eqref{eq:parity-projector}, and the last equality uses $\abs{f_{\bmalpha}(\yy)}=1$ from Eq.~\eqref{eq:falpha-def}.
Then we have
\begin{align}
\norm{g_{\bmalpha,\bmbeta}}_A
&=\sum_{\sss\in\ff^{m_X}}\abs{\widehat{f_{\bmalpha}}(\sss)}\abs{\widehat{f_{\bmbeta}}(\sss)}\nonumber\\
&\leq\qty(\sum_{\sss\in\ff^{m_X}}\abs{\widehat{f_{\bmalpha}}(\sss)}^2)^{1/2}
\qty(\sum_{\sss\in\ff^{m_X}}\abs{\widehat{f_{\bmbeta}}(\sss)}^2)^{1/2}\nonumber\\
&=1,
\label{eq:g-spectral-bound}
\end{align}
by the convolution theorem in Eq.~\eqref{eq:conv-thm}, the Cauchy-Schwarz inequality, and Eq.~\eqref{eq:f-l2}.
The remaining sections bound the right-hand side of Eq.~\eqref{eq:master-reduction} through the triangle inequality and the submultiplicativity of $\norm{\cdot}_A$.

\section{Polymer representation}
    \label{sec:polymer}

In this section, we derive a polymer representation of the function $g_{\bmalpha,\bmbeta}$ of Eq.~\eqref{eq:g-def}.
In Sec.~\ref{subsec:g-factorization}, we factor $g_{\bmalpha,\bmbeta}$ into two functions.
In Sec.~\ref{subsec:polymer-gas}, we introduce abstract polymer models, define the polymers and activities of our model, and obtain its partition function.
In Sec.~\ref{subsec:activity-bound}, we bound the activities in the Fourier $1$-norm for use in Sec.~\ref{sec:cluster}.

    \subsection{Factorization of \texorpdfstring{$g_{\bm\alpha,\bm\beta}$}{g}}
    \label{subsec:g-factorization}
To bound the right-hand side of Eq.~\eqref{eq:master-reduction}, we use a factorization that separates the dependence on $\bmalpha$.
Fix $\bmalpha,\bmbeta\in\ff^k$ and $\bmmu=\bmalpha+\bmbeta$ in Eq.~\eqref{eq:AA-double-sum}.
For $\xx\in\ff^{m_X}$ and a data qubit $j\in[n]$, 
consider the product of the factors at qubit $j$ in $f_{\bmalpha}(\yy)f_{\bmbeta}(\xx+\yy)$.
By Eq.~\eqref{eq:falpha-def}, when $(\bmmu O_X+\xx H_X)_j=1$, the two factors have opposite exponents and multiply to one.
When $(\bmmu O_X+\xx H_X)_j=0$, the exponents are equal and the product is $e^{\pm2\ii\theta_j}=\cos(2\theta_j)[1\pm\ii\tan(2\theta_j)]$.
To select qubits with equal exponents, define
\begin{equation}
\label{eq:Imu-def}
I_j^{\bmmu}(\xx)
\coloneqq
\frac{1+(-1)^{(\bmmu O_X)_j+(\xx H_X)_j}}{2}
\in\{0,1\}.
\end{equation}

Using these notations, define the functions
$C_{\bmmu}\colon\ff^{m_X}\to\mathbb{R}$ and
$\Xi_{\bmalpha|\bmmu}\colon\ff^{m_X}\to\cc$ by
\begin{equation}
\label{eq:Cmu-def}
C_{\bmmu}(\xx)
\coloneqq
\prod_{j=1}^{n}
\qty[\cos(2\theta_j)]^{I_j^{\bmmu}(\xx)}
\end{equation}
and
\begin{equation}
\label{Xi-def}
\Xi_{\bmalpha|\bmmu}(\xx)
\coloneqq
\sum_{\substack{\ccc\in\ff^n\\
\ccc H_X^\top=\bm0}}
\prod_{j\colon c_j=1}
\qty[
\ii\tan(2\theta_j)
(-1)^{(\bmalpha O_X)_j}
I_j^{\bmmu}(\xx)
].
\end{equation}

    With these definitions, we obtain the following factorization.
    \begin{lemma}
    \label{lem:g-factorized}
        For all $\bmalpha,\bmbeta\in\ff^k$ and all $\xx\in\ff^{m_X}$,
    \begin{equation}
    g_{\bm\alpha,\bm\beta}(\mathbf x)
    =C_{\bm\mu}(\mathbf x) \Xi_{\bm\alpha|\bm\mu}(\mathbf x),
    \end{equation}
    where $\bmmu=\bmalpha+\bmbeta$.
    \end{lemma}
    \begin{proof}
    For each data qubit $j\in[n]$, define $\eta_j\colon\ff^{m_X}\to\{\pm1\}$ by $\eta_j(\yy)\coloneqq (-1)^{(\yy H_X)_j}$.
    By Eq.~\eqref{eq:falpha-def} and $\eta_j(\xx+\yy)=\eta_j(\xx)\eta_j(\yy)$, the product of the factors at qubit $j$ in $f_{\bmalpha}(\yy)$ and $f_{\bmbeta}(\xx+\yy)$ is
\begin{align}
&\exp[\ii\theta_j(-1)^{(\bmalpha O_X)_j}\eta_j(\yy)]
\exp[\ii\theta_j(-1)^{(\bmbeta O_X)_j}\eta_j(\xx)\eta_j(\yy)]\nonumber\\
&=\exp\qty[
\ii\theta_j(-1)^{(\bmalpha O_X)_j}\eta_j(\yy)
\qty(1+(-1)^{(\bmmu O_X)_j}\eta_j(\xx))
]\nonumber\\
&=\exp[2\ii\theta_j(-1)^{(\bmalpha O_X)_j}
I^{\bmmu}_j(\xx) \eta_j(\yy)],
\end{align}
    where the first equality uses $(-1)^{(\bmbeta O_X)_j}=(-1)^{(\bmalpha O_X)_j}(-1)^{(\bmmu O_X)_j}$ by $\bmbeta=\bmalpha+\bmmu$, and the second equality uses $1+(-1)^{(\bmmu O_X)_j}\eta_j(\xx)=2I^{\bmmu}_j(\xx)$ by Eq.~\eqref{eq:Imu-def}.

Using $I_j^{\bmmu}(\xx)\in\{0,1\}$, we have
\begin{align}
\label{eq:local-phase}
&\exp\qty[
2\ii\theta_j(-1)^{(\bmalpha O_X)_j}
I_j^{\bmmu}(\xx)\eta_j(\yy)
]\nonumber\\
&=
[\cos(2\theta_j)]^{I_j^{\bmmu}(\xx)}
\left[
1+\ii\tan(2\theta_j)(-1)^{(\bmalpha O_X)_j}
I_j^{\bmmu}(\xx)\eta_j(\yy)
\right].
\end{align}
    
    Multiplying Eq.~\eqref{eq:local-phase} over $j\in[n]$ and using Eqs.~\eqref{eq:Cmu-def} and~\eqref{eq:g-def} gives
    \begin{align}
    \label{eq:g-intermediate}
        &g_{\bmalpha,\bmbeta}(\xx)\nonumber\\
        &=C_{\bmmu}(\xx)2^{-m_X}\nonumber\\
        &\quad\times \sum_{\yy\in \ff^{m_X}}\prod_{j=1}^n
        \qty[1+\ii\tan(2\theta_j)(-1)^{(\bmalpha O_X)_j}I_j^{\bmmu}(\xx)\eta_j(\yy)]\nonumber\\
        &=C_{\bmmu}(\xx)2^{-m_X}\nonumber\\
        &\quad\times \sum_{\substack{\yy\in\ff^{m_X}\\\ccc\in\ff^n}}\prod_{j\colon c_j=1}\qty[\ii\tan(2\theta_j)(-1)^{(\bmalpha O_X)_j}I_j^{\bmmu}(\xx)\eta_j(\yy)],
    \end{align}
    where the last equality expands the product over $j$.
    The factors other than $\eta_j(\yy)$ do not depend on $\yy$, and
    \begin{equation}
    \label{eq:eta-product}
    \prod_{j\colon c_j=1}\eta_j(\yy )
    =(-1)^{(\yy H_X)\cdot\mathbf c}
    =(-1)^{\yy \cdot(\mathbf cH_X^\top)}.
    \end{equation}
    Substituting Eq.~\eqref{eq:eta-product} into Eq.~\eqref{eq:g-intermediate} gives
\begin{align}
    &g_{\bmalpha,\bmbeta}(\xx)\nonumber\\
    &=C_{\bmmu}(\xx)\sum_{\ccc\in\ff^n}\prod_{j\colon c_j=1}\qty[\ii\tan(2\theta_j)(-1)^{(\bmalpha O_X)_j}I_j^{\bmmu}(\xx)]\nonumber\\
    &\quad\times2^{-m_X}\sum_{\yy\in\ff^{m_X}}(-1)^{\yy\cdot(\ccc H_X^\top)}\nonumber\\
    &=C_{\bmmu}(\xx)\Xi_{\bmalpha|\bmmu}(\xx),
\end{align}
    where the second equality uses Eqs.~\eqref{eq:parity-projector} and~\eqref{Xi-def}.
    \end{proof}

Lemma~\ref{lem:g-factorized} separates the factor that depends only on $\bmmu=\bmalpha+\bmbeta$ from the factor that depends on $\bmalpha$.
The next subsection expresses the latter factor $\Xi_{\bm\alpha|\bm\mu}$ as a polymer partition function.

 \subsection{Polymers and activities}
    \label{subsec:polymer-gas}
   In this subsection, we first introduce abstract polymer models and the adjacency graph.
   Then we define the polymers and the activities of our polymer model.
   Finally, we express $\Xi_{\bmalpha|\bmmu}$ of Eq.~\eqref{Xi-def} as the partition function of our model.

An \textit{abstract polymer model} consists of a finite set whose elements are called polymers, an activity assigned to each polymer, and a compatibility relation $\gamma\sim\gamma'$ on ordered pairs of polymers~\cite{Koteck__1986_polymer,friedli_velenik_2017_polymer}.
    The relation is required to be symmetric, i.e., $\gamma\sim\gamma'$ holds if $\gamma'\sim\gamma$ holds.
We say polymers $\gamma$ and $\gamma'$ are \textit{compatible} if $\gamma\sim\gamma'$ holds, and \textit{incompatible}, denoted by $\gamma\nsim\gamma'$, otherwise.
    Every polymer is required to be incompatible with itself.
    The \textit{partition function} of an abstract polymer model is the sum of the products of the activities over the finite sets of pairwise compatible polymers, including the empty set, whose contribution is one.
        To express the sum over the vectors $\ccc$ with $\ccc H_X^\top=\bm0$ in Eq.~\eqref{Xi-def} as the partition function of our model, we decompose the support of each such vector into connected components in a graph on the data qubits, and take these components as the polymers of our model.

To define these connected components, we introduce a graph, called the \textit{adjacency graph}, to describe which qubits may influence one another through a shared check.
The adjacency graph of the parity-check matrix $H_X$ is the graph $\calG_X=(V,E)$ with the vertex set $V\coloneqq [n]$ and the edge set $E\coloneqq \{\{j,j'\}\colon j,j^\prime\in V,j\neq j', \text{there exists $i\in [m_X]$ with $j,j'\in\mathrm{supp}(\mathbf{h}_X^{(i)})$}\}$, i.e., one vertex for each data qubit, with an edge between two distinct qubits when some $X$-type check acts on both~\cite{Gottesman2014constant,KovalevPryadko2013FaultTolerance}.
The degree of a vertex $j\in V$ is the number of edges containing it, $
\deg(j)\coloneqq\bigl\lvert\{j'\in V\colon\{j,j'\}\in E\}\bigr\rvert$, and $\Delta\coloneqq\max_{j\in V}\deg(j)$ denotes the maximum degree of $\calG_X$.
For $(r,c)$-QLDPC codes, the maximum degree $\Delta$ of $\calG_X$ is bounded by
\begin{equation}
    \Delta\leq c(r-1).
\label{eq: delta bound}
\end{equation}
This holds since each qubit belongs to at most $c$ generators and each such generator acts on at most $r-1$ other qubits.

A subset $S\subseteq V$ is \textit{connected} in $\calG_X$ if, for all $j,j'\in S$, there exist an integer $\ell\geq0$ and vertices $v_0,v_1,\ldots,v_\ell\in S$ such that
\begin{equation}
\label{eq:connected-def}
v_0=j,\quad v_\ell=j',\quad
\{v_{t-1},v_t\}\in E\quad(t=1,\ldots,\ell).
\end{equation}
For $j,j'\in S$, the relation that holds when Eq.~\eqref{eq:connected-def} is satisfied is an equivalence relation on $S$.
Its equivalence classes are called the $\textit{connected components}$ of $S$.
Hence, every subset $S\subseteq V$ is the disjoint union of its connected components.
Two disjoint subsets $\gamma,\gamma'\subseteq V$ are \textit{adjacent} in $\calG_X$ if $\{j,j'\}\in E$ for some $j\in\gamma$ and $j'\in \gamma'$, and \textit{nonadjacent} otherwise.
Distinct connected components $\gamma$ and $\gamma'$ of a subset are nonadjacent.
Conversely, if nonempty connected subsets $\gamma_1,\ldots,\gamma_m$ of $V$ are pairwise disjoint and nonadjacent, then they are the connected components of their union~\cite{West2001}.

    Next, we introduce zero-syndrome vectors, from which we will build polymers.
    We call a vector $\ccc\in\ff^n$ with $\ccc H_X^\top=\bm0$ a zero-syndrome vector.
    For a subset $\gamma\subseteq[n]$, let $\mathbf{1}_\gamma\in\ff^n$ denote its indicator vector, whose $j$-th entry is one if $j\in\gamma$ and zero otherwise.
    Using this notation, the following lemma reduces the zero-syndrome condition on $\ccc$ to the zero-syndrome condition on each connected component of its support.
    \begin{lemma}
    \label{lem:component-zero-syndrome}
        For any binary vector $\ccc\in\ff^n$, let $\gamma_1,\ldots, \gamma_M$ denote the connected components of $\mathrm{supp}(\ccc)$ in $\calG_X$.
        Then $\ccc$ is a zero-syndrome vector if and only if $\mathbf{1}_{\gamma_a}$ is a zero-syndrome vector for all $a\in[M]$.
    \end{lemma}

    \begin{proof}
        The connected components are disjoint and their union is $\mathrm{supp}(\ccc)$, so $\ccc=\sum_{a=1}^{M}\mathbf{1}_{\gamma_a}$ holds, and we have for every row $\mathbf{h}_X^{(i)}$ of $H_X$,
        \begin{equation}
        \label{eq:syndrome-split}
            \mathbf{h}_X^{(i)}\cdot \ccc=\sum_{a=1}^M \mathbf{h}_X^{(i)}\cdot\one_{\gamma_a}.
        \end{equation}
    If $\one_{\gamma_a} H_X^\top=\bm0$ for every $a$, then every term on the right-hand side of Eq.~\eqref{eq:syndrome-split} is zero, so $\ccc H_X^\top=\bm0$.

    Conversely, for a fixed row $\mathbf{h}_X^{(i)}$, at most one connected component $\gamma_a$ satisfies $\mathrm{supp}(\mathbf{h}_X^{(i)})\cap\gamma_a\neq\emptyset$.
Suppose instead that $\mathrm{supp}(\mathbf{h}_X^{(i)})\cap\gamma_a\neq\emptyset$ and $\mathrm{supp}(\mathbf{h}_X^{(i)})\cap \gamma_b\neq \emptyset$ for two different connected components $\gamma_a$ and $\gamma_b$, and take data qubits $j\in\mathrm{supp}(\mathbf{h}_X^{(i)})\cap \gamma_a$ and $j'\in\mathrm{supp}(\mathbf{h}_X^{(i)})\cap \gamma_b$.
    The row $\mathbf{h}_X^{(i)}$ is nonzero at both $j$ and $j'$, and $j\neq j'$ since the components are disjoint, so $\{j,j'\}\in E$ by the definition of the adjacency graph, and $\gamma_a$ and $\gamma_b$ are adjacent, contradicting that distinct connected components are nonadjacent.

    Since $\mathbf{h}_X^{(i)}\cdot \one_{\gamma_a}$ is the parity of $\abs{\mathrm{supp}(\mathbf h_X^{(i)})\cap\gamma_a}$, at most one term on the right-hand side of Eq.~\eqref{eq:syndrome-split} is nonzero.
    If $\mathbf cH_X^\top=\bm 0$, the left-hand side of Eq.~\eqref{eq:syndrome-split} is zero, and a sum over $\ff$ with at most one nonzero term is zero only when every term is zero.
    Hence $\mathbf h_X^{(i)}\cdot\one_{\gamma_a}=0$ for every $a$, and since the row is arbitrary, $\one_{\gamma_a}H_X^\top=\bm 0$ for every $a$.

    \end{proof}
    
    We now specify the polymers, motivated by Lemma~\ref{lem:component-zero-syndrome} and by the decompositions of errors into connected subsets used in analyses of QLDPC codes~\cite{KovalevPryadko2013FaultTolerance,Gottesman2014constant}.
    A polymer is a nonempty subset $\gamma\subseteq[n]$ that is connected in $\calG_X$ and whose indicator vector satisfies $\onegamma H_X^\top=0$.
    We say two polymers are compatible if they are disjoint and nonadjacent in $\calG_X$.
    For each $\bmalpha,\bmmu\in\ff^k$ and each polymer $\gamma$, define the activity of $\gamma$ as the function $w_{\gamma}^{\bmalpha|\bmmu}\colon \ff^{m_X}\to\cc$ given by
    \begin{equation}
    \label{eq:activity-def}
        w_\gamma^{\bmalpha|\bmmu}(\xx)\coloneqq \prod_{j\in\gamma}\qty[\ii\tan(2\theta_j)(-1)^{(\bmalpha O_X)_j}I^{\bmmu}_j(\xx)].
    \end{equation}
    Polymer activities defined as functions are also used in Refs.~\cite{Dimock_2000,Lohmann_2015}.
    The polymers, the activities, and the compatibility relation defined above form an abstract polymer model, since the relation is symmetric in the two polymers and every polymer is nonempty and hence not disjoint from itself.

    Finally, we show that the zero-syndrome vectors correspond one to one to the finite sets of pairwise compatible polymers, and then rewrite the sum over $\ccc$ in Eq.~\eqref{Xi-def} as the partition function.
        For a zero-syndrome vector $\ccc$, let $\gamma_1,\ldots,\gamma_M$ denote the connected components of $\mathrm{supp}(\ccc)$ in $\calG_X$, so that $\ccc=\sum_{a=1}^{M}\mathbf{1}_{\gamma_a}$ and each $\mathbf{1}_{\gamma_a}$ is a zero-syndrome vector by Lemma~\ref{lem:component-zero-syndrome}.
    Each component is nonempty and connected, hence a polymer, and distinct components are disjoint and nonadjacent, hence compatible.
     Thus $\{\gamma_1,\ldots,\gamma_M\}$ is a finite set of pairwise compatible polymers.
    Conversely, for every finite set $\Gamma$ of pairwise compatible polymers, the vector $\ccc=\sum_{\gamma\in\Gamma}\mathbf{1}_\gamma$ is a zero-syndrome vector since each $\mathbf{1}_\gamma$ is a zero-syndrome vector, the support of $\ccc$ is the union of the polymers in $\Gamma$ since the polymers in $\Gamma$ are disjoint, and the connected components of $\mathrm{supp}(\ccc)$ are the elements of $\Gamma$ since these polymers are connected, disjoint, and nonadjacent.
    These two assignments are inverse to each other, with the empty set assigned to $\ccc=\bm0$.
    For a finite set $\Gamma$ of pairwise compatible polymers and the corresponding zero-syndrome vector $\ccc=\sum_{\gamma\in\Gamma}\mathbf{1}_\gamma$, the polymers in $\Gamma$ are disjoint and the union of the polymers in $\Gamma$ is $\mathrm{supp}(\ccc)$, so the term of $\ccc$ in Eq.~\eqref{Xi-def} is the product over $\gamma\in\Gamma$ of the activities of Eq.~\eqref{eq:activity-def}.
    Therefore $\Xi_{\bmalpha|\bmmu}$ is the partition function of this model,
    \begin{equation}
        \label{eq:polymer-gas}
        \Xi_{\bmalpha|\bmmu}
        =\sum_{\Gamma\colon \text{compatible}}\ \prod_{\gamma\in\Gamma}
        w^{\bmalpha|\bmmu}_\gamma,
    \end{equation}
    where the sum runs over the \emph{compatible sets} $\Gamma$, i.e., the finite sets of pairwise compatible polymers, and the empty set contributes the constant function one.

  \subsection{Bounding the activities}
    \label{subsec:activity-bound}
    We now bound the activities $w_{\gamma}^{\bmalpha|\bmmu}$ in terms of the Fourier $1$-norm of Eq.~\eqref{eq:spectral-norm}.
    For $j\in[n]$, let $\mathbf{h}_j\in\ff^{m_X}$ be the $j$-th row of $H_X^\top$, so $(\mathbf xH_X)_j=\mathbf x\cdot\mathbf h_j$.
    Then Eq.~\eqref{eq:Imu-def} is expressed as
    \begin{equation}
    I^{\bm\mu}_j(\xx)
    =\frac{1+(-1)^{(\bm\mu O_X)_j}\chi_{\mathbf h_j}(\xx)}{2},
    \label{eq:Imu-character}
    \end{equation}
    with the character $\chi_{\mathbf h_j}$ of Eq.~\eqref{eq:character}.
    The constant function one is the character $\chi_{\bm 0}$, and every character $\chi_{\bmxi}$ satisfies $\norm{\chi_{\bmxi}}_A=1$.
    The triangle inequality of Eq.~\eqref{eq:spectral-triangle} applied to Eq.~\eqref{eq:Imu-character} gives
    \begin{equation}
        \label{eq:I-wiener-bound}
    \norm{I^{\bm\mu}_j}_A
    \leq\frac{\norm{\chi_{\bm 0}}_A+\norm{\chi_{\mathbf h_j}}_A}{2}
    =1.
    \end{equation}
    By the submultiplicativity of \eqref{eq:spectral-submultiplicative}, Eqs.~\eqref{eq:I-wiener-bound} and~\eqref{eq: delta} give
    \begin{equation}
    \norm{w^{\bm\alpha|\bm\mu}_\gamma}_A
    \leq\prod_{j\in\gamma}\abs{\tan(2\theta_j)}\norm{I^{\bm\mu}_j}_A
    \leq\delta^{\abs{\gamma}},
    \label{eq:activity-bound}
    \end{equation}
    where $\abs{\gamma}$ denotes the number of data qubits in $\gamma$.
    Each data qubit contained in a polymer costs a factor of at most $\delta$ in the Fourier $1$-norm.

The bound in Eq.~\eqref{eq:activity-bound} gives an estimate for each activity, but summing these estimates over compatible sets may give an upper bound that grows exponentially with $n$.
In the following section, we use the cluster expansion to rewrite the sum over compatible sets in Eq.~\eqref{eq:polymer-gas} as the exponential of a series over clusters.
We bound the differences between these series before taking the exponential, rather than applying the triangle inequality directly to the compatible sets.

    \section{Cluster expansion}
    \label{sec:cluster}
    In this section, we prove Proposition~\ref{prop:average-bound} below, which bounds $\norm{g_{\bmalpha,\bmalpha+\bmmu}-g_{\mathrm{ave}|\bmmu}}_A$ by a constant multiple of $ne^{-bd_Z}$, and then use Eq.~\eqref{eq:master-reduction} to complete the proof of the main theorem.
    In Sec.~\ref{subsec:cluster-setup}, we introduce the cluster expansion of the abstract polymer model.
    In Sec.~\ref{subsec:kp-verify}, we state Theorem~\ref{thm:kp-rooted} and verify the Koteck\'y-Preiss condition for our activities in Lemma~\ref{lem:kp-verified}.
    In Sec.~\ref{subsec:average-proof}, we prove Proposition~\ref{prop:average-bound} and complete the proof of the main theorem.

     \subsection{Cluster expansion of abstract polymer models}
    \label{subsec:cluster-setup}
    The \textit{cluster expansion} for abstract polymer models~\cite{Koteck__1986_polymer,friedli_velenik_2017_polymer,Bissacot_2010_polymer} expresses the logarithm of the partition function.
    To express the logarithm of the partition function $\Xi_{\bmalpha|\bmmu}$ of Eq.~\eqref{eq:polymer-gas}, we introduce a tuple of polymers, the incompatibility graph, clusters, and the Ursell function as follows.
    Let $(\gamma_1,\ldots, \gamma_m)$ be a tuple of polymers, $m\geq 1$, with repetitions allowed.
    The \textit{incompatibility graph} of $(\gamma_1,\ldots, \gamma_m)$ is the graph with vertex set $\{1,\ldots, m\}$ and edge set $\{\{a,a'\}\colon 1\leq a <a'\leq m,\gamma_a\nsim \gamma_{a'}\}$.
    A \textit{cluster} is a tuple whose incompatibility graph is connected.
    A \textit{spanning subgraph} of a graph is a subgraph whose vertex set is the whole vertex set of the graph, so a spanning subgraph is specified by a subset of the edge set.
    The \textit{Ursell function} of a tuple $(\gamma_1, \ldots,\gamma_m)$ is
    \begin{equation}
    \label{eq:ursel-def}\phi^{\mathrm{T}}(\gamma_1, \ldots,\gamma_m)\coloneqq \sum_{G}(-1)^{|E(G)|},\end{equation}
    where $G$ runs over all connected spanning subgraphs of the incompatibility graph of the tuple and $E(G)$ denotes the edge set of $G$.
    For $m=1$, the incompatibility graph consists of a single vertex and no edges. 
    Its unique spanning subgraph is this graph itself, which is connected, so $\phi^{\mathrm{T}}=1$.
    If the tuple is not a cluster, no such subgraph exists and $\phi^{\mathrm{T}}=0$.

    For $\xx\in\ff^{m_X}$, consider the cluster series associated with Eq.~\eqref{eq:polymer-gas},
    \begin{equation}
    \label{eq:cluster-series}
        L_{\bmalpha|\bmmu}(\xx)\coloneqq \sum_{m\geq 1}\frac{1}{m!}\sum_{\gamma_1,\ldots,\gamma_m}\phi^{\mathrm T}(\gamma_1,\ldots, \gamma_m)\prod_{a=1}^m w_{\gamma_a}^{\bmalpha|\bmmu}(\xx),
    \end{equation}
    where the inner sum runs over ordered tuples of $m$ polymers and repetitions are allowed.
    Since $\phi^{\mathrm T}$ vanishes for tuples that are not clusters, only the clusters contribute to Eq.~\eqref{eq:cluster-series}.
    For a function $h\colon\ff^{m_X}\to\cc$, let $\exp(h)$ denote the function with $\exp(h)(\xx)\coloneqq e^{h(\xx)}$ for every $\xx\in\ff^{m_X}$.

     \subsection{Convergence under the Koteck\'y-Preiss condition}
    \label{subsec:kp-verify}
    In this subsection, we prove absolute convergence of the cluster series in Eq.~\eqref{eq:cluster-series} at every $\xx$ and derive the bound on the sum over clusters used in Sec.~\ref{subsec:average-proof}. 
    We first state the Koteck\'y-Preiss theorem for complex activities and then verify its condition using the activity bound of Eq.~\eqref{eq:activity-bound}.
\begin{theorem}[Koteck\'y-Preiss theorem~\cite{Koteck__1986_polymer,friedli_velenik_2017_polymer}]
\label{thm:kp-rooted}
Let $z_\gamma$ be a nonnegative real number for every polymer $\gamma$, and assume that every polymer $\gamma$ satisfies the Koteck\'y-Preiss condition
\begin{equation}
\label{eq:kp-functional}
\sum_{\gamma'\nsim\gamma}z_{\gamma'}e^{\abs{\gamma'}}\leq\abs{\gamma}.
\end{equation}
Then, for every family of complex activities $u=(u_\gamma)_\gamma$ satisfying $\abs{u_\gamma}\leq z_\gamma$, the cluster series
\begin{equation}
\label{eq:scalar-cluster-series}
\ell(u)\coloneqq\sum_{m\geq1}\frac1{m!}
\sum_{\gamma_1,\ldots,\gamma_m}
\phi^{\mathrm T}(\gamma_1,\ldots,\gamma_m)
\prod_{a=1}^m u_{\gamma_a}
\end{equation}
converges absolutely and satisfies
\begin{equation}
\label{eq:scalar-partition-function}
\Xi(u)\coloneqq\sum_{\Gamma}\ \prod_{\gamma\in\Gamma}u_\gamma
=\exp\qty(\ell(u)),
\end{equation}
where $\Gamma$ runs over all finite sets of pairwise compatible polymers, including the empty set, whose contribution is one.
In addition, every polymer $\gamma_1$ satisfies
\begin{equation}
\label{eq:kp-rooted}
\sum_{m\geq1}\frac1{(m-1)!}\sum_{\gamma_2,\ldots,\gamma_m}
\abs{\phi^{\mathrm T}(\gamma_1,\gamma_2,\ldots,\gamma_m)}\prod_{a=2}^m z_{\gamma_a}
\leq e^{\abs{\gamma_1}},
\end{equation}
with the $m=1$ term equal to one, where the inner sum runs over ordered tuples of polymers with repetitions allowed.
\end{theorem}

Theorem~\ref{thm:kp-rooted} follows from Proposition~5.3 and Theorem~5.4 of Ref.~\cite{friedli_velenik_2017_polymer} for our compatibility relation with the positive function $a(\gamma)=\abs{\gamma}$.
The Ursell function in that reference equals $\phi^{\mathrm T}(\gamma_1,\ldots,\gamma_m)/m!$.
Applying Theorem~5.4 of Ref.~\cite{friedli_velenik_2017_polymer} to the activities $z_\gamma$ gives Eq.~\eqref{eq:kp-rooted}.

For our activities, let $b>0$ and define the majorants
\begin{equation}
\label{eq:zgamma-def}
z_\gamma=(\delta e^b)^{\abs{\gamma}}.
\end{equation}
They are majorants of the activities since $\norm{w_\gamma^{\bmalpha|\bmmu}}_A\leq\delta^{\abs{\gamma}}=e^{-b|\gamma|}z_{\gamma}\leq z_\gamma$ by Eq.~\eqref{eq:activity-bound}.
The factor $e^{-b\abs{\gamma}}$ is introduced to obtain the factor $e^{-bd_Z}$ in Sec.~\ref{subsec:average-proof}.
Verifying Eq.~\eqref{eq:kp-functional} for these majorants requires counting the polymers of each size that contain a fixed qubit.
The same counting also bounds the sum $\sum_\gamma z_\gamma e^{\abs{\gamma}}$ used below.

By Lemma~2 of Ref.~\cite{Gottesman2014constant}, the number of connected vertex subsets of size $m$ that contain a fixed qubit is at most $(e\Delta)^{m-1}$.
Since every polymer is a connected subset of $\calG_X$, the number of size-$m$ polymers containing a fixed qubit is at most $(e\Delta)^{m-1}$.
This factor depends only on the maximum degree $\Delta$, not on $n$.
By Eq.~\eqref{eq: delta bound}, the row and column weight bounds of $(r,c)$-QLDPC code bound $\Delta$ by the constant $c(r-1)$, independently of $n$.

The following lemma shows that Eq.~\eqref{eq:kp-condition} implies Eq.~\eqref{eq:kp-functional}, which is the assumption of Theorem~\ref{thm:kp-rooted}, and also gives the bound in Eq.~\eqref{eq:polymer-sum-bound}.
\begin{lemma}
\label{lem:kp-verified}
Let $\Delta\geq1$ be an upper bound on the maximum degree of $\calG_X$.
Let $b>0$, and let $\delta\geq0$ be given by Eq.~\eqref{eq: delta}.
For every polymer $\gamma$, let $z_\gamma=(\delta e^b)^{\abs{\gamma}}$, and define
\begin{equation}
\label{eq:xb-def}
x_b=\delta e^{b+1}.
\end{equation}
Assume
\begin{equation}
\label{eq:kp-condition}
x_b<\frac1{(e+1)\Delta+1}.
\end{equation}
This condition is equivalent to $\delta<e^{-(b+1)}/[(e+1)\Delta+1]$.
Then every polymer $\gamma$ satisfies 
\begin{equation}
\sum_{\gamma'\nsim\gamma}z_{\gamma'}e^{\abs{\gamma'}}\leq\abs{\gamma},
\end{equation}
which is Eq.~\eqref{eq:kp-functional}, and
\begin{equation}
\label{eq:polymer-sum-bound}
\sum_{\gamma}z_\gamma e^{\abs{\gamma}}\leq n\frac{x_b}{1-e\Delta x_b}.
\end{equation}
\end{lemma}

\begin{proof}
We first verify Eq.~\eqref{eq:kp-functional}.
A polymer $\gamma'$ incompatible with a polymer $\gamma$ shares a qubit with $\gamma$, or contains a qubit $j'$ with $\{j,j'\}\in E$ for some $j\in\gamma$, so $\gamma'$ contains one of the qubits $j'$ with $j'\in\gamma$ or $\{j,j'\}\in E$ for some $j\in\gamma$, whose number is at most $(\Delta+1)\abs{\gamma}$ since $\deg(j)\leq\Delta$ for every $j\in\gamma$.
Summing over these qubits and using the bound $\qty(e\Delta)^{m-1}$ on the number of size-$m$ polymers containing a fixed qubit gives
\begin{align}
\label{eq:kp-check}
\sum_{\gamma'\nsim\gamma}z_{\gamma'}e^{\abs{\gamma'}}
&\leq(\Delta+1)\abs{\gamma}\sum_{m\geq1}(e\Delta)^{m-1}\qty(\delta e^{b+1})^m
\nonumber\\
&=(\Delta+1)\abs{\gamma}\frac{x_b}{1-e\Delta x_b}\nonumber\\
&<\abs{\gamma},
\end{align}
where Eq.~\eqref{eq:kp-condition} gives $e\Delta x_b<1$ and $(\Delta+1)x_b/(1-e\Delta x_b)<1$.

We next prove Eq.~\eqref{eq:polymer-sum-bound}.
Every polymer contains at least one qubit, so the same counting gives
\begin{equation}
\sum_{\gamma}z_\gamma e^{\abs{\gamma}}
\leq\sum_{j=1}^{n}\sum_{m\geq1}(e\Delta)^{m-1}\qty(\delta e^{b+1})^{m}
=n\frac{x_b}{1-e\Delta x_b}.
\end{equation}
\end{proof}

Under Eq.~\eqref{eq:kp-condition}, Lemma~\ref{lem:kp-verified} shows that the majorants satisfy the assumption of Theorem~\ref{thm:kp-rooted}.
For every $\bmalpha,\bmmu\in\ff^k$ and $\xx\in\ff^{m_X}$, Eqs.~\eqref{eq:pointwise-bound}, \eqref{eq:activity-bound}, and~\eqref{eq:zgamma-def} give
\begin{equation*}
\abs{w_\gamma^{\bmalpha|\bmmu}(\xx)}
\leq\norm{w_\gamma^{\bmalpha|\bmmu}}_A
\leq\delta^{\abs{\gamma}}\leq z_\gamma.
\end{equation*}
With $u_\gamma=w_\gamma^{\bmalpha|\bmmu}(\xx)$, the cluster series $\ell(u)$ of Eq.~\eqref{eq:scalar-cluster-series} equals $L_{\bmalpha|\bmmu}(\xx)$ of Eq.~\eqref{eq:cluster-series}, and the partition function $\Xi(u)$ of Eq.~\eqref{eq:scalar-partition-function} equals $\Xi_{\bmalpha|\bmmu}(\xx)$ by Eq.~\eqref{eq:polymer-gas}.
Thus Theorem~\ref{thm:kp-rooted} shows that Eq.~\eqref{eq:cluster-series} converges absolutely at every $\xx$ and that its sum satisfies
\begin{equation}
\label{eq:Xi-exp-L}
\Xi_{\bmalpha|\bmmu}(\xx)=\exp\qty(L_{\bmalpha|\bmmu}(\xx)).
\end{equation}
Since $\ff^{m_X}$ is finite, pointwise convergence of functions on $\ff^{m_X}$ implies convergence of every Fourier coefficient and hence convergence in the Fourier $1$-norm.

Grouping the tuples by their first entry $\gamma_1$ and using $1/m!\leq1/(m-1)!$ gives
\begin{align}
&\sum_{m\geq1}\frac1{m!}
\sum_{\gamma_1,\ldots,\gamma_m}
\abs{\phi^{\mathrm T}(\gamma_1,\ldots,\gamma_m)}
\prod_{a=1}^m z_{\gamma_a}\nonumber\\
&\leq\sum_{\gamma_1}z_{\gamma_1}
\sum_{m\geq1}\frac1{(m-1)!}
\sum_{\gamma_2,\ldots,\gamma_m}
\abs{\phi^{\mathrm T}(\gamma_1,\ldots,\gamma_m)}
\prod_{a=2}^m z_{\gamma_a}\nonumber\\
&\leq\sum_\gamma z_\gamma e^{\abs{\gamma}}\nonumber\\
&\leq n\frac{x_b}{1-e\Delta x_b},
\label{eq:cs-sum-bound}
\end{align}
where the second inequality uses Eq.~\eqref{eq:kp-rooted} and the third uses Eq.~\eqref{eq:polymer-sum-bound}.
Thus Eq.~\eqref{eq:cs-sum-bound} bounds the sum over clusters by a constant multiple of $n$.
The remaining step compares the cluster series for different $\bmalpha$ and retains only the clusters with a nontrivial effect on the encoded information, using $d_Z$.

\subsection{Proof of the main theorem}
\label{subsec:average-proof}

For a tuple $(\gamma_1,\ldots,\gamma_m)$ of polymers with repetitions allowed, we call the logical class $(\sum_{a=1}^m {\onegamma}_a)O_X^{\top}\in\ff^{k}$ of the sum of the indicator vectors of its polymers the \textit{total logical class} of the tuple, and $\sum_{a=1}^m |\gamma_a|$ its \textit{total size}.
For each $\bmalpha'\neq\bmalpha$, the terms whose total logical class is zero cancel in the difference $L_{\bmalpha'|\bmmu}-L_{\bmalpha|\bmmu}$.
For every remaining tuple, the sum in $\ff^n$ of the indicator vectors of its polymers has zero syndrome and nonzero logical class, so the total size of the tuple is at least $d_Z$.

The following proposition combines this distance bound
with Eq.~\eqref{eq:cs-sum-bound} to obtain the required
bound on $\norm{g_{\bmalpha,\bmalpha+\bmmu}-g_{\mathrm{ave}|\bmmu}}_A$.
\begin{proposition}
\label{prop:average-bound}
Let $\Delta\geq1$ be an upper bound on the maximum degree of $\calG_X$, and let $b>0$ satisfy Eq.~\eqref{eq:kp-condition}.
Then, for all $\bmalpha,\bmmu\in\ff^k$,
\begin{equation}
\label{eq:average-difference-bound}
\norm{g_{\bmalpha,\bmalpha+\bmmu}-g_{\mathrm{ave}|\bmmu}}_A
\leq\qty(1-2^{-k})C_1 ne^{-bd_Z},
\end{equation}
with $C_1=\max\qty{e^{C_0}-1,2}$ and $C_0=2x_b/(1-e\Delta x_b)$, where $x_b$ is given by Eq.~\eqref{eq:xb-def}.
\end{proposition}

\begin{proof}
Fix $\bmalpha,\bmmu\in\ff^k$. We first bound $L_{\bmalpha'|\bmmu}-L_{\bmalpha|\bmmu}$ for an arbitrary $\bmalpha'\in\ff^k$.
For every $\bmalpha'\in\ff^k$ and every polymer $\gamma$, Eq.~\eqref{eq:activity-def} gives
\begin{equation}
\label{eq:single-polymer-character}
w_\gamma^{\bmalpha'|\bmmu}
=(-1)^{(\bmalpha'+\bmalpha)\cdot(\onegamma O_X^\top)}
w_\gamma^{\bmalpha|\bmmu}.
\end{equation}
Multiplying Eq.~\eqref{eq:single-polymer-character} over the polymers of a tuple $(\gamma_1,\ldots,\gamma_m)$ gives
\begin{align}
\label{eq:tuple-character}
\prod_{a=1}^m w_{\gamma_a}^{\bmalpha'|\bmmu}
={}&(-1)^{(\bmalpha'+\bmalpha)\cdot
\qty[\qty(\sum_{a=1}^m\one_{\gamma_a})O_X^\top]}
\prod_{a=1}^m w_{\gamma_a}^{\bmalpha|\bmmu},
\end{align}
where the sum of the indicator vectors is taken in $\ff^n$.
By absolute convergence of Eq.~\eqref{eq:cluster-series}, the two series may be subtracted term by term.
Equations~\eqref{eq:cluster-series} and~\eqref{eq:tuple-character} give
\begin{align}
\label{eq:L-difference-cluster-sum}
&L_{\bmalpha'|\bmmu}-L_{\bmalpha|\bmmu}\nonumber\\
&=\sum_{m\geq1}\frac1{m!}
\sum_{\gamma_1,\ldots,\gamma_m}
\phi^{\mathrm T}(\gamma_1,\ldots,\gamma_m)\nonumber\\
&\quad\times\qty[(-1)^{(\bmalpha'+\bmalpha)\cdot
\qty[\qty(\sum_{a=1}^m\one_{\gamma_a})O_X^\top]}-1]
\prod_{a=1}^m w_{\gamma_a}^{\bmalpha|\bmmu}.
\end{align}
Only the tuples satisfying
\begin{equation}
(\bmalpha'+\bmalpha)\cdot
\qty[\qty(\sum_{a=1}^m\one_{\gamma_a})O_X^\top]=1
\end{equation}
contribute to Eq.~\eqref{eq:L-difference-cluster-sum}.
In particular, every contributing tuple has nonzero total logical class.
For every tuple with nonzero total logical class, define $\ccc\coloneqq\sum_{a=1}^m\one_{\gamma_a}\in\ff^n$.
Every polymer satisfies $\one_{\gamma_a}H_X^\top=\bm0$, so
\begin{equation}
\ccc H_X^\top=\sum_{a=1}^m\one_{\gamma_a}H_X^\top=\bm0.
\end{equation}
The nonzero total logical class gives $\ccc O_X^\top\neq\bm0$, and hence Eq.~\eqref{eq:dz-def} gives $\wt(\ccc)\geq d_Z$.
The support of $\ccc$ is contained in the union of the polymers in the tuple, and the union has at most $\sum_{a=1}^m \abs{\gamma_a}$ elements, so
\begin{equation}
\label{eq:tuple-size-distance}
\sum_{a=1}^m\abs{\gamma_a}\geq\wt(\ccc)\geq d_Z.
\end{equation}

By Eq.~\eqref{eq:activity-bound} and the definition $z_\gamma=(\delta e^b)^{\abs{\gamma}}$ of Eq.~\eqref{eq:zgamma-def},
\begin{equation}
\norm{w_{\gamma_a}^{\bmalpha|\bmmu}}_A
\leq\delta^{\abs{\gamma_a}}
=z_{\gamma_a}e^{-b\abs{\gamma_a}}.
\end{equation}
We now apply the triangle inequality of Eq.~\eqref{eq:spectral-triangle} and the submultiplicativity of Eq.~\eqref{eq:spectral-submultiplicative} to the finite sums with $m\leq M$ in Eq.~\eqref{eq:L-difference-cluster-sum}, extending the sum to every tuple with nonzero total logical class.
These finite sums converge pointwise to $L_{\bmalpha'|\bmmu}-L_{\bmalpha|\bmmu}$, and hence in the Fourier $1$-norm because $\ff^{m_X}$ is finite. 
Taking $M\to\infty$ therefore gives
\begin{align}
&\norm{L_{\bmalpha'|\bmmu}-L_{\bmalpha|\bmmu}}_A\nonumber\\
&\leq2
\sum_{m\geq1}\frac1{m!}
\sum_{\substack{\gamma_1,\ldots,\gamma_m\\
\qty(\sum_{a=1}^m\one_{\gamma_a})O_X^\top\neq\bm0}}
\abs{\phi^{\mathrm T}(\gamma_1,\ldots,\gamma_m)}\nonumber\\
&\qquad\times\prod_{a=1}^m z_{\gamma_a}e^{-b\abs{\gamma_a}}\nonumber\\
&\leq2e^{-bd_Z}
\sum_{m\geq1}\frac1{m!}
\sum_{\gamma_1,\ldots,\gamma_m}
\abs{\phi^{\mathrm T}(\gamma_1,\ldots,\gamma_m)}
\prod_{a=1}^m z_{\gamma_a}\nonumber\\
&\leq2e^{-bd_Z}\sum_\gamma z_\gamma e^{\abs{\gamma}}\nonumber\\
&\leq\frac{2x_b}{1-e\Delta x_b} ne^{-bd_Z}
\eqqcolon2\varepsilon,
\label{eq:L-difference-bound}
\end{align}
where in the second inequality Eq.~\eqref{eq:tuple-size-distance} gives $\prod_{a=1}^m e^{-b\abs{\gamma_a}}=e^{-b\sum_{a=1}^m\abs{\gamma_a}}\leq e^{-bd_Z}$ and the condition
$(\sum_{a=1}^m\one_{\gamma_a})O_X^\top\neq\bm0$
can then be removed because all terms are nonnegative, the third inequality uses Eq.~\eqref{eq:cs-sum-bound}, and the fourth uses Eq.~\eqref{eq:polymer-sum-bound}.

We next convert Eq.~\eqref{eq:L-difference-bound} into a bound on $g_{\bmalpha,\bmalpha+\bmmu}-g_{\mathrm{ave}|\bmmu}$ through Lemma~\ref{lem:g-factorized} and Eq.~\eqref{eq:Xi-exp-L}.
For all $h\colon\ff^{m_X}\to\cc$, the series $\sum_{m\geq1}h(\xx)^m/m!$ converges to $e^{h(\xx)}-1$ at every $\xx$.
The finite sums $\sum_{m=1}^M h^m/m!$ therefore converge to $\exp(h)-1$ in the Fourier $1$-norm because $\ff^{m_X}$ is finite.
Applying Eqs.~\eqref{eq:spectral-homogeneous}, \eqref{eq:spectral-triangle}, and~\eqref{eq:spectral-submultiplicative} to the finite sums $\sum_{m=1}^M h^m/m!$ and taking $M\to\infty$ gives
\begin{equation}
\label{eq:exp-norm-bound}
\norm{\exp(h)-1}_A
\leq\sum_{m\geq1}\frac{\norm{h}_A^m}{m!}
=e^{\norm{h}_A}-1.
\end{equation}

The definition of $g_{\mathrm{ave}|\bmmu}$ gives
\begin{align}
&g_{\bmalpha,\bmalpha+\bmmu}-g_{\mathrm{ave}|\bmmu}\nonumber\\
&=2^{-k}\sum_{\substack{\bmalpha'\in\ff^k\\\bmalpha'\neq\bmalpha}}
\qty(g_{\bmalpha,\bmalpha+\bmmu}-g_{\bmalpha',\bmalpha'+\bmmu})\nonumber\\
&=2^{-k}\sum_{\substack{\bmalpha'\in\ff^k\\\bmalpha'\neq\bmalpha}}
C_{\bmmu}\qty[\exp\qty(L_{\bmalpha|\bmmu})-\exp\qty(L_{\bmalpha'|\bmmu})]\nonumber\\
&=2^{-k}\sum_{\substack{\bmalpha'\in\ff^k\\\bmalpha'\neq\bmalpha}}
g_{\bmalpha,\bmalpha+\bmmu}
\qty[1-\exp\qty(L_{\bmalpha'|\bmmu}-L_{\bmalpha|\bmmu})],
\label{eq:g-average-L}
\end{align}
where the term with $\bmalpha'=\bmalpha$ is zero in the first equality, the second equality uses Lemma~\ref{lem:g-factorized} and Eq.~\eqref{eq:Xi-exp-L} for both $\bmalpha$ and $\bmalpha'$, and the last equality uses $\exp(h+h')=\exp(h)\exp(h')$ for all functions $h, h'\colon \ff^{m_X}\to \cc$.
Equation~\eqref{eq:g-average-L} gives
\begin{align}
&\norm{g_{\bmalpha,\bmalpha+\bmmu}-g_{\mathrm{ave}|\bmmu}}_A\nonumber\\
&\leq2^{-k}\sum_{\substack{\bmalpha'\in\ff^k\\\bmalpha'\neq\bmalpha}}
\norm{g_{\bmalpha,\bmalpha+\bmmu}}_A
\norm{1-\exp\qty(L_{\bmalpha'|\bmmu}-L_{\bmalpha|\bmmu})}_A\nonumber\\
&\leq2^{-k}\sum_{\substack{\bmalpha'\in\ff^k\\\bmalpha'\neq\bmalpha}}
\qty[e^{\norm{L_{\bmalpha'|\bmmu}-L_{\bmalpha|\bmmu}}_A}-1]\nonumber\\
&\leq2^{-k}\qty(2^k-1)\qty(e^{2\varepsilon}-1)\nonumber\\
&=\qty(1-2^{-k})\qty(e^{C_0ne^{-bd_Z}}-1),
\label{eq:g-average-exp-bound}
\end{align}
where the first inequality uses Eqs.~\eqref{eq:spectral-homogeneous}-\eqref{eq:spectral-submultiplicative}, the second inequality uses Eq.~\eqref{eq:g-spectral-bound} and Eq.~\eqref{eq:exp-norm-bound}, and the third inequality uses Eq.~\eqref{eq:L-difference-bound}.

The first equality in Eq.~\eqref{eq:g-average-L} and Eq.~\eqref{eq:g-spectral-bound} give
\begin{equation}
\label{eq:g-average-trivial-bound}
\norm{g_{\bmalpha,\bmalpha+\bmmu}-g_{\mathrm{ave}|\bmmu}}_A
\leq2\qty(1-2^{-k}).
\end{equation}
Let $t=ne^{-bd_Z}$.
If $t\leq1$, the right-hand side of Eq.~\eqref{eq:g-average-exp-bound} is at most $\qty(1-2^{-k})(e^{C_0}-1)t$, since $e^{C_0t}-1$ is convex in $t$ and vanishes at $t=0$.
If $t\geq1$, the right-hand side of Eq.~\eqref{eq:g-average-trivial-bound} is at most $2\qty(1-2^{-k})t$.
In both cases Eq.~\eqref{eq:average-difference-bound} holds with $C_1=\max\qty{e^{C_0}-1,2}$.
\end{proof}

The two factors in Eq.~\eqref{eq:average-difference-bound} have different origins.
The restriction in Eq.~\eqref{eq:L-difference-cluster-sum} gives Eq.~\eqref{eq:tuple-size-distance} and hence the factor $e^{-bd_Z}$.
The factor $n$ comes from the sum over the possible locations of a polymer in Eq.~\eqref{eq:polymer-sum-bound}.

Finally, we prove Theorem~\ref{thm:threshold theorem}.
\begin{proof}[Proof of Theorem~\ref{thm:threshold theorem}]
Fix an arbitrary $i\geq1$ and omit the index $i$ from the notation.
Then $n$, $k$, $d$, $d_Z$, $\Delta$, $\bm\theta$, $\delta$, and $F_e^{\mathrm{ML}}$ refer to the fixed code $\mathcal{Q}_i$.
The constants $\Delta_0$, $\delta_*$, and $k_0$ of Theorem~\ref{thm:threshold theorem} do not depend on $i$, and $\delta\leq\delta_*$ by Eq.~\eqref{eq:delta-theorem}.
Since $\Delta_0=c(r-1)\geq1$ for $r\geq2$ and $c\geq 1$, and $\Delta\leq\Delta_0$ by Eq.~\eqref{eq: delta bound}, Lemma~\ref{lem:kp-verified} and Proposition~\ref{prop:average-bound} can be applied with $\Delta$ replaced by $\Delta_0$.
Choose
\begin{equation}
\label{eq: b-def}
b\coloneqq\frac12\log\qty(\frac{\delta_{\mathrm{th}}}{\delta_*}),
\end{equation}
as in Eq.~\eqref{eq:constant-b}.
The value depends only on $\Delta_0$ and $\delta_*$ and does not depend on the code size.
By Eqs.~\eqref{eq:xb-def} and~\eqref{eq:constant-b},
\begin{align}
\label{eq:b-check}
x_b&=\delta e^{b+1}\nonumber\\
&\leq\delta_* e^{b+1}\nonumber\\
&=e\sqrt{\delta_*\cdot\delta_{\mathrm{th}}}\nonumber\\
&<e\delta_{\mathrm{th}}=\frac{1}{(e+1)\Delta_0+1},
\end{align}
where we use $\delta\leq \delta_*<\delta_{\mathrm{th}}$ and Eq.~\eqref{eq:delta-th-def}.
Thus Eq.~\eqref{eq:kp-condition} holds with $\Delta$ replaced by $\Delta_0$, and Proposition~\ref{prop:average-bound} together with Eq.~\eqref{eq:master-reduction} gives
\begin{equation}
\label{eq:thm-proof-step1}
1-F_e^{\mathrm{ML}}\leq\binom{q}{2}\qty(1-q^{-1})C_1 n e^{-bd_Z},
\end{equation}
where $q=2^k$.
Equation~\eqref{eq:b-check} gives $x_b<1/[(e+1)\Delta_0+1]$, so
\begin{equation}
C_0=\frac{2x_b}{1-e\Delta_0x_b}<\frac{2}{\Delta_0+1}\leq1
\end{equation}
by $\Delta_0\geq1$.
Hence $e^{C_0}-1\leq e-1<2$, and $C_1=\max\qty{e^{C_0}-1,2}=2$.
Substituting $C_1=2$ into Eq.~\eqref{eq:thm-proof-step1} gives
\begin{equation}
1-F_e^{\mathrm{ML}}\leq(q-1)^2n e^{-bd_Z},
\end{equation}
which is Eq.~\eqref{eq:codewise-bound} because $q=2^k$.
Using $q=2^k\leq2^{k_0}$ and $d\leq d_Z$ further gives
\begin{align}
1-F^{\mathrm{ML}}_{e}
&\leq\qty(2^{k_0}-1)^2ne^{-bd_Z}\nonumber\\
&\leq\qty(2^{k_0}-1)^2n e^{-bd},
\end{align}
where the prefactor $(2^{k_0}-1)^2$ is the constant $C$ of Eq.~\eqref{eq:constant-C}, which does not depend on $i$.
Thus Eq.~\eqref{eq:main-bound} holds for all $i\geq1$.
Finally, suppose that $d_i/\log n_i\to\infty$.
Since $b>0$ does not depend on $i$, for all sufficiently large $i$, we have $n_ie^{-bd_i}\to0$.
Equation~\eqref{eq:main-bound} therefore gives $1-F^{\mathrm{ML}}_{e,i}\to0$.
\end{proof}

\section{Conclusion}
\label{sec:conclusion}
We have proved a finite-size bound in the code-capacity setting for CSS QLDPC codes with a bounded number of logical qubits under coherent $Z$-rotation errors and the maximum-likelihood Pauli recovery.
More precisely, the entanglement infidelity is suppressed exponentially in the code distance up to a prefactor linear in the number of data qubits.
For code families whose distance grows superlogarithmically in the number of physical qubits, the finite-size bound gives a positive lower bound on the coherent-error threshold, one determined only by the sparsity of the codes.

To obtain the bound, we derived an exact expression for the entanglement fidelity under this recovery and bounded the infidelity using pairwise overlaps between distinct logical classes.
Expressing the transition amplitudes through the phases of the computational-basis states under the error, we obtained a representation of these overlaps through Fourier coefficients and used the restriction to distinct logical classes before taking absolute values. The resulting Fourier 1-norm bound retains the relevant interference and is controlled by the cluster expansion~\cite{Koteck__1986_polymer,friedli_velenik_2017_polymer}.
In the difference between the cluster series for $\bmalpha$ and $\bmalpha'$, clusters with zero total logical class cancel exactly, while the polymers in every surviving cluster have total size at least $d_Z$.
This gives the exponential factor in the finite-size bound.

The result concerns ideal syndrome measurement, the maximum-likelihood Pauli recovery, and QLDPC codes with a bounded number of logical qubits, which suggests the following directions for future work.
One of these directions is to prove the existence of a threshold under circuit-level coherent errors.
Another is to prove the existence of a threshold under more practical decoders, such as minimum-weight decoders~\cite{Dennis_2002,Higgott2022PyMatching,Gottesman2014constant,Takada2026Doubly}.
A third is to extend the result to high-rate QLDPC families with sublinear distance~\cite{TillichZemor2014Quantum,Leverrier2015Expander,Hastings2021Fiber,Bravyi_2024}, because the bound contains the factor $\qty(2^k-1)^2$ of Eq.~\eqref{eq:codewise-bound}.
This requires controlling the dependence on the number of logical qubits $k$ or avoiding the pairwise reduction.
More broadly, Fourier analysis that converts interfering amplitudes into a form controlled by the cluster expansion may offer a new statistical-mechanical perspective on coherent errors and contribute to other problems in many-body physics~\cite{zhang2025stabilitymixedstatephasesweak,Wild_2023,Yang2008Quantum}.

\begin{acknowledgments}
This work was supported by JST [Moonshot R\&D][Grant Number JPMJMS256E].
We acknowledge the use of ChatGPT 5.6 for additional verification of our results.
All ideas and proofs were developed by the authors.
\end{acknowledgments}

\bibliography{bibliography}

\end{document}